\documentclass[12pt,draftcls,peerreviewca,onecolumn]{IEEEtran}

\usepackage{etoolbox}
\makeatletter
\patchcmd{\@makecaption}
  {\scshape}
  {}
  {}
  {}
\makeatletter
\patchcmd{\@makecaption}
  {\\}
  {.\ }
  {}
  {}
\makeatother
\usepackage{amsfonts}
\usepackage{mathrsfs}
\usepackage{amsfonts}
\usepackage{amssymb}
\usepackage{graphicx}
\usepackage{epsfig}
\usepackage{psfrag}
\usepackage{amsmath}
\usepackage{array}
\usepackage{cases}
\usepackage{subfigure}
\usepackage{cite,graphicx,amsmath,amssymb,color}
\usepackage{commath}
\usepackage{anysize}
\usepackage{algorithmic}
\usepackage{algorithm}
\usepackage{epstopdf}
\DeclareMathOperator*{\argmax}{argmax}

\newcommand{\dis}{\stackrel{d}{\sim}}
\newcommand{\eqla}{\stackrel{(a)}{=}}
\newcommand{\eqlb}{\stackrel{(b)}{=}}

\renewcommand{\mod}[0]{\text{ mod }}

\newtheorem{Thm}{Theorem}
\newtheorem{Lem}{Lemma}
\newtheorem{Cor}{Corollary}
\newtheorem{Def}{Definition}

\newtheorem{Prob}{Problem}
\newtheorem{Rem}{Remark}

\IEEEoverridecommandlockouts

\begin{document}

\setcounter{page}{1}
\title{Joint and Competitive Caching Designs in Large-Scale Multi-Tier Wireless Multicasting Networks}
\author{\authorblockN{Zitian Wang\thanks{Z. Wang, Z. Cao and Y. Cui are with the Department of  Electronic Engineering, Shanghai Jiao Tong University, China. Y. Yang is with the Intel Deutschland GmbH, Germany. This paper was submitted in part to IEEE GLOBECOM 2017.},\quad Zhehan Cao,\quad Ying Cui,\quad Yang Yang}}

\maketitle
\thispagestyle{headings}


\begin{abstract}
Caching and multicasting are two promising methods to support massive content delivery  in multi-tier wireless networks.
In this paper, we consider a  random caching  and multicasting scheme with caching distributions in the two tiers as design parameters, to achieve efficient content dissemination in a two-tier large-scale cache-enabled wireless multicasting network. First, we  derive  tractable expressions for  the successful transmission probabilities in the general region as well as the high SNR and high user density region, respectively, utilizing tools from stochastic geometry.
Then, for the case of a single operator for the two tiers,
we formulate  the optimal joint caching design problem to maximize the successful transmission probability in the asymptotic region, which is nonconvex in general. By using the block successive approximate optimization technique, we develop an iterative algorithm, which is shown to converge to a stationary point.
Next, for the case of two  different operators, one for each tier,
we formulate the competitive caching design game where each tier maximizes its successful transmission probability in the asymptotic region. We show that the game has a unique Nash equilibrium (NE) and develop an iterative algorithm, which is shown to converge to  the NE under a mild condition.
Finally, by numerical simulations, we show that the proposed designs achieve significant gains over existing schemes.
\end{abstract}

\begin{keywords}
Cache, multicast, multi-tier wireless network,
stochastic geometry, optimization,
game theory, Nash equilibrium
\end{keywords}
\newpage

\section{Introduction}

The rapid proliferation of smart mobile devices has triggered
an unprecedented growth of the global mobile data traffic.
Multi-tier wireless networks have been proposed as an effective way
to meet the dramatic traffic growth by deploying different tiers of point of attachments (POAs), e.g., base stations (BSs) or access points (APs) together, to provide better time or frequency reuse. In general, there are two scenarios, depending on whether different tiers are managed by the same operator. One typical example for the scenario of the same operator is deploying short range small-BSs together with traditional macro-BSs, i.e., heterogeneous wireless networks (HetNets). One typical example for the scenario of different operators is deploying IEEE 802.11 APs of different owners.
To further reduce the load of the core network, caching at POAs in multi-tier wireless networks is recognized as a promising  approach.

Caching in cache-enabled multi-tier wireless networks for the case of the same operator is considered in many works.
Cache-enabled multi-tier wireless networks with fixed topologies are considered in some of them.
For example, in \cite{Shanmugam13,LiTWC15,cui2016optimal}, the authors  consider the optimal content placement at small-BSs to minimize the expected downloading time for files at the macro-BS in a single macro-cell  with multiple small-cells.
Note that \cite{Shanmugam13,LiTWC15,cui2016optimal} do not capture  the stochastic natures of  channel fading  and geographic locations of  POAs and users, and the obtained results in \cite{Shanmugam13,LiTWC15,cui2016optimal} may not be applied to real networks. To address these limitations, large-scale cache-enabled multi-tier wireless networks are considered in some other works, using tools from stochastic geometry.
For example, 
in \cite{EURASIP15Debbah,LiuYangICC16,Yang16}, the authors consider caching the most popular files at each small-BS in large-scale cache-enabled small-cell networks or HetNets. 
In \cite{QuekTWC16}, the authors propose  a partition-based combined caching design in a large-scale cluster-centric small-cell network. 
In \cite{TamoorComLett16} and \cite{DBLP:journals/corr/Tamoor-ul-Hassan15}, the authors consider random caching of a uniform distribution at small-BSs in a  large-scale cache-enabled HetNet and a large-scale cache-enabled small-cell network, respectively. 
In \cite{zhang2016energy}, each macro-BS caches the most popular files and each small-BS randomly caches popular files in a large-scale cache-enabled HetNet.
Note that the focuses in \cite{EURASIP15Debbah,LiuYangICC16,Yang16,QuekTWC16,TamoorComLett16,DBLP:journals/corr/Tamoor-ul-Hassan15,zhang2016energy} are only on performance analysis of some simple caching designs, which may not provide  performance guarantee.
In \cite{chen2016probabilistic,wen2017random,li2016optimization,wen2016cache}, the authors consider random caching and
focus on the analysis  and optimization of the probability that the signal-to-interference plus noise ratio (SINR) of a typical user is above a threshold, in a large-scale cache-enabled multi-tier wireless network. 
In \cite{chen2016probabilistic}, the authors consider two architectures (i.e., an always-on architecture and a dynamic on-off architecture), and formulate the optimization problem for each architecture, which is convex. For each problem, the closed-form optimal solution is obtained. 
In \cite{wen2017random}, the authors consider two cooperative transmission schemes, and formulate the optimization problem under each scheme, which is nonconvex in the general case. For each problem, a stationary point is obtained using the standard gradient projection method. 
For \cite{li2016optimization,wen2016cache}, in a special case where all tiers have the same threshold, the optimization problem is convex and the optimal solution is obtained; in the general case, the problem is nonconvex. In~\cite{wen2016cache}, the nonconvex problem is simplified to a convex one  and the optimal solution to the simplified convex problem is used as a sub-optimal solution to the original nonconvex problem.
In \cite{serbetci2016optimal}, the authors propose a random caching design, and focus on the maximization of the cache hit probability. The optimization problem is  convex and the closed-form optimal solution is obtained.
Note that 
\cite{chen2016probabilistic,wen2017random,serbetci2016optimal,li2016optimization,wen2016cache} focus only on the typical user and  do not consider the resource sharing among multiple users.


Some works consider competitive caching among different POAs, using game theory.
For instance, in \cite{tan2016femtocaching}, the authors consider an Exact Potential Game among cache-enabled femto-BSs where each femto-BS maximizes the expected
number of its served users, prove the existence of Nash equilibrium (NE) and propose a convergent algorithm to obtain a NE. In \cite{kim2017ultra}, the authors consider a mean-field game among cache-enabled small-BSs where each small-BS minimizes its long run average cost, and obtain the unique mean field equilibrium.
For example, in \cite{li2016commercial,poularakis2014framework,shen2016stackelberg}, the authors consider Stackelberg games among content providers and network operators. Specifically, the content providers rent part of the network resources from the network operators for content delivery to get payment from  users. Note that in \cite{tan2016femtocaching,poularakis2014framework,shen2016stackelberg}, the authors consider cache-enabled wireless networks with fixed topologies. 
In \cite{kim2017ultra,li2016commercial}, large-scale cache-enabled wireless networks are considered;  in \cite{kim2017ultra}, a large-scale cache-enabled single-tier wireless network is considered; in \cite{li2016commercial}, the authors consider a large-scale cache-enabled multi-tier wireless network, but do not provide a convergent algorithm to find the Stackelberg equilibrium.


On the other hand, enabling multicast service at POAs in multi-tier wireless networks is  an efficient way to deliver popular contents to multiple requesters simultaneously by effectively utilizing the broadcast nature of the wireless medium.
In our previous work \cite{cui2017analysis}, we consider analysis and optimization of a hybrid caching design and a corresponding  multicasting design in a large-scale cache-enabled HetNet. 
The hybrid caching design requires  the files stored at macro-BSs and pico-BSs to be nonoverlapping and the files stored at all macro-BSs  to be identical. Thus, the spatial file diversity provided by the hybrid caching design is limited, which may cause network performance degradation at some system parameters. 
In our previous work \cite{cui2016analysis}, we consider analysis and optimization of a random caching design and a corresponding multicasting design  in a large-scale cache-enabled single-tier network. The proposed random caching design in \cite{cui2016analysis} can offer high spatial file diversity, ensuring good network performance over a wide range of  system parameters, but can not be directly applied to HetNets.

In summary,
further studies are required to facilitate the design of practical cache-enabled multi-tier wireless multicasting networks for massive content dissemination. 
In this paper, we consider a  random caching and multicasting design with caching distributions in the two tiers as the design parameters to provide high spatial file diversity. We derive  tractable expressions for  the successful transmission probabilities in the general region as well as the high SNR and high user density region (i.e., the asymptotic region), respectively, utilizing tools from stochastic geometry.
Our main contributions are summarized below.                                                                                                                                                                                                                                                                                                                                                       
\begin{itemize}
\item For the case of a single operator for the two tiers,
we formulate  the optimal joint caching design problem to maximize the successful transmission probability in the asymptotic region, which is a nonconvex problem in general. By using the block successive approximate optimization technique\cite{razaviyayn2013unified}, we develop  an iterative algorithm to obtain a stationary point. Specifically, by carefully choosing an approximation function, we obtain the closed-form optimal solution to the approximate optimization problem in each iteration.
In addition, in the special case of the same cache size, we develop a low-complexity algorithm to obtain a globally optimal solution by extending the method in\cite{wen2016cache}.
\item For the case of two  different operators, one for each tier,
we formulate the competitive caching design game where each tier maximizes its successful transmission probability in the asymptotic region. We show that the game has a unique NE and develop an iterative algorithm to obtain the NE. In general, it is quite difficult to guarantee that an iterative algorithm can converge to the NE of a game, especially for a large-scale wireless network.   
By carefully analyze structural properties of the competitive caching design game, we provide a convergence condition for the proposed iterative algorithm, which holds in most practical scenarios.
\item Finally, by numerical simulations, we show that the proposed designs achieve significant gains over  existing schemes in terms of successful transmission probability and complexity. We also show the caching probabilities of the proposed designs, revealing that the proposed designs offer high spatial file diversity. 
\end{itemize}

\section{System Model}\label{sec:netmodel}
\subsection{Network Model}

We consider a general large-scale two-tier downlink network consisting of two tiers of POAs, e.g., BSs or APs, as shown in Fig.~\ref{fig:system}. The two tiers can be managed by a single operator (e.g., HetNet with BSs being POAs) or by two different operators (e.g., IEEE 802.11 APs of two owners).\footnote{The network model considered in this paper is similar to that in\cite{cui2017analysis}. But here, we consider a random caching design which is more general and includes the hybrid caching design in \cite{cui2017analysis} as a special case. In addition, different from \cite{li2016optimization, wen2016cache}, we specify the random caching design by the caching probabilities of file combinations, so as to investigate the file load distribution and the impact of multicasting.} The locations of the POAs in tier $ 1 $ and tier $ 2 $ are spatially distributed as two independent homogeneous Poisson point processes (PPPs) $\Phi_{1}$ and $\Phi_{2}$ with densities $\lambda_{1}$ and $\lambda_{2}$, respectively. The locations of the users are also distributed as an independent homogeneous PPP $\Phi_{u}$ with density $\lambda_{u}$.
Each POA in the $j$th tier  has one transmit antenna with   transmission power $P_j$, where $j=1,2$. 
For notational convenience, we  define $\sigma_1\triangleq\frac{P_1}{P_2}$ and $\sigma_2\triangleq\frac{P_2}{P_1}$.
Each user has one receive antenna. All POAs are operating on the same frequency band with a bandwidth  $W$ (Hz). 
Consider a discrete-time system with time being slotted and study one slot of the network.
Both path loss and small-scale fading are considered: for path loss, a transmitted signal from either tier with distance $D$ is attenuated by a factor $D^{-\alpha}$, where $\alpha>2$ is the path loss exponent~\cite{li2016optimization,wen2016cache}; for small-scale fading, Rayleigh fading channels are adopted~\cite{WCOM13Andrews}.

Let $\mathcal N\triangleq \{1,2,\cdots, N\}$ denote the set of $N$ files  in the two-tier network.
For ease of illustration, assume that all  files  have the same size.
File popularity is assumed to be identical among all users.  Each user randomly  requests one file, which is file $n\in \mathcal N$ with probability $a_n\in (0,1)$, where $\sum_{n\in \mathcal N}a_n=1$.  Thus, the file popularity distribution is given by $\mathbf a\triangleq (a_n)_{n\in \mathcal N }$, which is assumed to be known  apriori\cite{li2016optimization, wen2016cache}.\footnote{Note that file popularity evolves at a slower timescale and learning methodologies can be employed to track the evolution
of file popularity over time.}
In addition, without loss of generality (w.l.o.g.),  assume  $a_{1}> a_{2}>\ldots> a_{N}$.
The  two-tier network   consists of  cache-enabled POAs. In the $j$th tier, each POA is equipped with a cache of size $K_j<N$ to store different popular files out of $N$. 

We say every $K_j$ different files form a combination. Thus, there are in total $I_j\triangleq \binom{N}{K_j}$ different combinations, each with $K_j$ different files.
Let $\mathcal I_j\triangleq \{1,2,\cdots, I_j\}$ denote the set of $I_j$ combinations. Combination $i\in \mathcal I_j$ can be characterized by an $N$-dimensional vector $\mathbf x_{j,i}\triangleq (x_{j,i,n})_{n\in \mathcal N}$, where $x_{j,i,n}=1$ if file $n$ is included in combination $i$ of tier $j$ and $x_{j,i,n}=0$ otherwise. Note that there are  $K_j$ 1's  in  each $\mathbf x_{j,i}$.
Denote $\mathcal N_{j,i}\triangleq \{n\in\mathcal N: x_{j,i,n}=1\}$ as the set of $K_j$ files contained in combination $i$ of tier $j$.

\begin{figure}[t]
\begin{center}
\includegraphics[width=14cm]{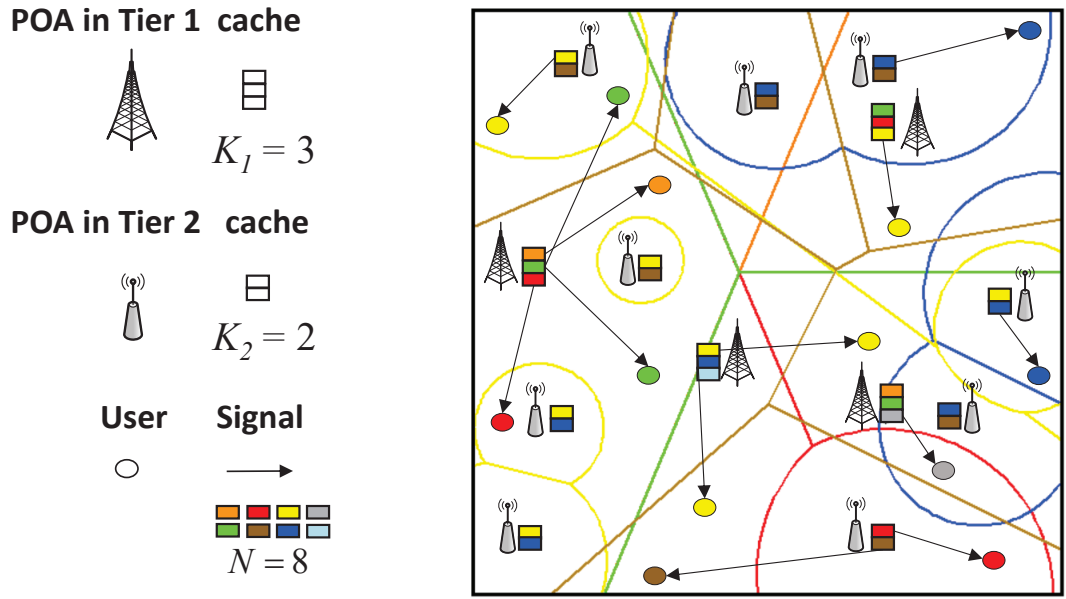}
\end{center}
\caption{\small{Network model.
Each file $n\in \mathcal N$ corresponds to a Voronoi tessellation (in the same color as the file),  determined  by the locations and transmission powers of all POAs storing this file.}}
\label{fig:system}
\end{figure}

\subsection{Caching}
To provide high spatial file diversity, we consider a random caching design in the cache-enabled two-tier network where the caching distributions in the two tiers may be different, as illustrated  in Fig.~\ref{fig:system}.
The probability that combination $ i\in\mathcal I_j $ is stored in each POA of tier $ j $ is $ p_{j,i} $, 
where $p_{j,i}$ satisfies
\begin{align}
0\leq p_{j,i}\leq1, \, i\in \mathcal I_j,\quad \sum_{i\in \mathcal I_j}p_{j,i}=1.\label{eqn:cache-constr-indiv}
\end{align}
A random caching design in the tier $j$ is specified by the caching distribution $\mathbf p_j\triangleq (p_{j,i})_{ i\in \mathcal I_j}$. Let $\mathcal I_{j,n} \triangleq \{i \in \mathcal I_j : x_{j,i,n}=1\}$ denote the set of $I_{j,n} \triangleq \binom{N-1}{K_j-1}$ combinations containing file $n$.
Let
\begin{align}
T_{j,n} \triangleq \sum_{i\in\mathcal {I}_{j,n}}p_{j,i}, \ n \in \mathcal N \label{eqn:def-T-n}
\end{align}
denote the probability that file $n$ is stored at a POA in the $j$th tier. Therefore, the random caching design in the large-scale cache-enabled two-tier network is fully specified by the design parameters $\left(\mathbf p_1,\mathbf p_2\right)$. 
In this paper, we focus on serving cached files at POAs to get first-order insights into the design of cache-enabled wireless networks, as in \cite{Andrews11,cui2016analysis,cui2017analysis,wen2016cache,li2016optimization}. POAs may serve uncached files through other service machanisms, the investigation of which is beyond the scope of this paper.
\begin{Rem}
Note that the random caching design considered in this paper is a generalization of the caching design where the most popular files are stored at each POA and the hybrid caching design proposed in \cite{cui2017analysis}. In particular, by choosing the design parameters $(\mathbf p_1, \mathbf p_2)$ such that $T_{j,n}=1$ for all $n=1,2,\cdots,K_j$ and $T_{j,n}=0$ for all $n=K_j+1,K_j+2,\cdots,N$, where $j=1,2$, the proposed random caching design turns to the design caching the most popular files\cite{QuekTWC16,DBLP:journals/corr/Tamoor-ul-Hassan15}. In addition, by choosing the design paprameters $(\mathbf p_1, \mathbf p_2)$ in a certain manner, the proposed random caching design can reflect identical caching in the $ 1 $st tier, random caching in the $ 2 $nd tier and nonoverlapping caching across the two tiers, and hence incorporate the hybrid caching design in\cite{cui2017analysis} as a special case. Therefore, by carefully designing $(\mathbf p_1, \mathbf p_2)$, the proposed random caching design can achieve better performance than the two designs. Later, we shall see the advantage of the proposed random design in Section~\ref{Sec:simu}.
\end{Rem}

\subsection{Multicasting}

Consider a user requesting file $n$. If file $n$ is not stored in any tier, the user will not be served. Otherwise adopt the following user association rule:
i) if file $n$ is stored only in the $ j $th tier, the user is associated with the nearest POA in the $j$th tier storing file $ n $; ii) if file $n$ is stored in both tiers, the user is associated with the POA which stores file $n$ and provides the maximum long-term average received power (RP) among all the POAs \cite{wen2016cache,li2016optimization}.


\begin{Rem}
Note that the content-centric user association considered in this paper is a generalization of the content-centric user association in \cite{cui2017analysis}. In particular, Case ii) is not included in \cite{cui2017analysis} due to the nonoverlapping caching constraint in~\cite{cui2017analysis}.
\end{Rem}

We consider multicasting
in the large-scale cache-enabled  two-tier network. Consider a POA serving requests for $k$ different files. Then, it transmits each of the $k$  files only once to concurrently serve  users requesting the same file, at a rate
$\tau $ (bit/second) and over $\frac{1}{k}$  of the total bandwidth $W$ using frequency division multiple access (FDMA). 
As a matter of fact, both multicast and unicast may happen (with different probabilities). Without loss of generality, as in \cite{cui2017analysis}, we refer to this transmission as multicast.
Note that, by avoiding transmitting  the same file multiple times to multiple users, this content-centric multicast
can  improve the efficiency of the utilization of the wireless medium and reduce the load of the wireless network, compared to the traditional connection-based unicast\cite{WCOM13Andrews}.
From the above illustration, we can see that
the proposed multicasting design is also affected by the proposed caching design. Therefore,
the design parameters  $\left(\mathbf p_1,\mathbf p_2\right)$ affect the performance of the random caching and multicasting design.

\subsection{Performance Metric}\label{Sec:perf}


In this paper, we study w.l.o.g. the performance of a typical user $u_0$, which is located at the origin.
Suppose $u_0$ requests file $n$. Let $j_0$ denote the index of the tier with which $u_0$ is associated, and let $\overline j_0$ denote the other tier.  Let $\ell_0\in\Phi_{j_0}$ denote the index of the serving POA of $u_0$.
We denote $D_{j,\ell,0}$ and $h_{j,\ell,0}\dis \mathcal{CN}\left(0,1\right)$ as the distance and the small-scale channel between POA $\ell\in\Phi_j$ and $u_{0}$, respectively.
We assume the complex additive white Gaussian noise of power $N_0$ (evaluated over the entire frequency band) at $u_0$.
For analytical tractability, as in \cite{cui2016analysis} and \cite{cui2017analysis}, we assume all POAs are active for serving their own users. This corresponds to the worst-case interference strength for the typical user.\footnote{The performance obtained under this assumption provides a lower bound on the performance of the practical network where some void POAs may be shut down.}
When $u_{0}$ requests file $n$ and file $n$ is transmitted by POA $\ell_0$, the  SINR of $u_{0}$ is given by\footnote{The bandwidth for serving $u_0$ is random, and affects the signal, interference and noise power experienced at $u_0$ in the same manner, i.e., linearly. Thus, we can use the signal and noise power over the whole frequency band in calculating ${\rm SINR}_{n,0}$\cite{cui2017analysis}.}
\begin{align}
{\rm SINR}_{n,0}=\frac{{D_{j_0,\ell_0,0}^{-\alpha_{j_0}}}\left|h_{j_0,\ell_0,0}\right|^{2}}{\sum_{\ell\in\Phi_{j_0}\backslash \ell_0}D_{j_0,\ell,0}^{-\alpha_{j_0}}\left|h_{j_0,\ell,0}\right|^{2}+\sum_{\ell\in\Phi_{\overline j_0}}D_{\overline j_0,\ell,0}^{-\alpha_{\overline j_0}}\left|h_{\overline j_0,\ell,0}\right|^{2}\frac{P_{\overline j_0}}{P_{j_0}}+\frac{N_{0}}{P_{j_0}}}, \ n\in\mathcal N.
\label{eqn:SINR}
\end{align}
Note that, as in\cite{QuekTWC16,cui2017analysis}, the transmitted symbols of file $n$ from POA $\ell_0$ are treated as the desired signal, while the transmitted symbols of file $n$ from other POAs are regarded as interference.\footnote{The received signals from all the POAs transmitting file $n$ may not be perfectly synchronized due to the large difference in distances from these POAs to $u_0$\cite{chae2015cooperative}.}
When $T_{j,n}>0$ (i.e., $u_0$ may be associated with tier $j$), let $K_{j,n,0}\in \{1,\cdots, K_j\}$ denote the number of different cached files requested by the users associated with POA $\ell_0\in\Phi_{j}$. 
Note that $K_{j,n,0}$ is a discrete random variable, whose probability mass function (p.m.f.) depends on $\mathbf a$, $\lambda_u$ and the design parameters $\left(\mathbf p_1,\mathbf p_2\right)$.

The file can be decoded correctly at $u_0$ if the channel capacity between BS $\ell_0$ and $u_0$ is greater than or equal to $\tau$. 
Requesters are mostly concerned about whether their desired files can be successfully received. Therefore, we adopt the probability that a randomly requested file by $u_0$ is successfully transmitted, referred to as the successful transmission probability,
as the network performance metric\cite{cui2017analysis}.
Let $A_{j,n}(\mathbf p_j,\mathbf p_{\overline j})$ denote the probability that $u_0$ requesting file $n$ is associated with tier $j$. By total probability theorem, the successful transmission probability under the considered scheme is
\begin{align}
q\left(\mathbf p_1,\mathbf p_2\right)=& \underbrace{ \sum_{n\in\mathcal N}a_n A_{1,n}\left(\mathbf p_1,\mathbf p_2\right){\rm Pr}\left[\frac{W}{K_{1,n,0}}\log_{2}\left(1+{\rm SINR}_{n,0}\right)\geq\tau \ \middle| \ \text{$j_0=1$}\right] }_{\triangleq q_1(\mathbf p_1, \mathbf p_2) } \nonumber\\
&+\underbrace{ \sum_{n\in\mathcal N}a_n A_{2,n}\left(\mathbf p_2,\mathbf p_1\right){\rm Pr}\left[\frac{W}{K_{2,n,0}}\log_{2}\left(1+{\rm SINR}_{n,0}\right)\geq\tau \ \middle| \ \text{$j_0=2$} \right] }_{\triangleq q_2(\mathbf p_1, \mathbf p_2) } \label{eqn:succ-prob-def}.
\end{align}
where $q_j(\mathbf p_j,\mathbf p_{\overline j})$ represents the probability that a randomly requested file by $u_0$ is successfully transmitted from a POA in tier $j$, also referred to as the successful transmission probability of tier $j$.\footnote{Note that the expression of the successful transmission probability in \eqref{eqn:succ-prob-def} is different from the performance metrics in \cite{cui2017analysis} and \cite{WCOM13Andrews}.}

\section{Performance Analysis}


In this section, we first analyze  the successful transmission probability in the general region. Then, we analyze the successful transmission probability in the asymptotic region. 

\subsection{Performance Analysis in General Region}

In this subsection, we  analyze the successful transmission probability in the general region (i.e., the general SNR and general user density region), using tools from stochastic geometry.
First, the user association probability $A_{j,n}(\mathbf p_j, \mathbf p_{\overline j})$ can be found in\cite{wen2016cache,li2016optimization} and is provided here for completeness:
\begin{align}
A_{j,n}(\mathbf p_{j},\mathbf p_{\overline j}) = \frac{\lambda_jT_{j,n}}{\lambda_jT_{j,n} + \lambda_{\overline j}T_{\overline j,n}\left(\frac{P_{\overline j}}{P_j}\right)^{ \frac{2}{\alpha} }} \triangleq A_{j,n}(T_{j,n},T_{\overline j,n}). \label{eqn:user-association-prob}
\end{align}

File load $K_{j,n,0}$ and SINR ${\rm SINR}_{n,0}$ are correlated in a complex manner in general, as POAs with larger association regions have higher file load and lower SINR (due to larger user to POA distances) \cite{AndrewsTWCOffloading14}.  For the tractability of the analysis, as in\cite{cui2016analysis,cui2017analysis,AndrewsTWCOffloading14}, the dependence is ignored, i.e.,
\begin{align}
&{\rm Pr}\left[\frac{W}{K_{j,n,0}}\log_{2}\left(1+{\rm SINR}_{n,0}\right)\geq\tau  \ \middle| \ j_0=j \right] \nonumber\\
\approx&\sum_{k=1}^{K_j}{\rm Pr}\left[K_{j,n,0}=k \ \middle| \ j=j_0\right]{\rm Pr}\left[{\rm SINR}_{n,0}\geq \left(2^{\frac{k\tau}{W}}-1\right) \  \middle| \ j_0=j\right],  \ j=1,2,\ n\in\mathcal N. \label{eqn:performance-given-n-j}
\end{align}
To obtain the conditional p.m.f. of $K_{j,n,0}$ given $j_0=j$ by generalizing the methods for calculating the p.m.f. of file load in\cite{cui2017analysis},
we need the probability density function (p.d.f.) of the size of the Voronoi cell of BS $\ell_{0}$ w.r.t. file $ m\in\mathcal{N}_{j,i,-n} $  when  $\ell_{0}$ contains combination $i\in \mathcal I_{j,n}$, where $\mathcal{N}_{j,i,-n}\triangleq \mathcal N_{j,i} \setminus \{n\}$.  However, since this p.d.f. is very complex and still unknown, we adopt the widely used approach in the existing literature\cite{AndrewsTWCOffloading14,WCOM13Andrews,cui2016analysis,cui2017analysis} and approximate this p.d.f. based on a tractable approximation of the  p.d.f. of the size of the Voronoi cell to which a randomly chosen user belongs\cite{SGcellsize13}. 
Under this approximation, the conditional p.m.f. of $K_{j,n,0}$ is given in the following lemma.
\begin{Lem} [Conditional p.m.f. of $K_{j,n,0}$] The conditional p.m.f. of $K_{j,n,0}$  given $j_0=j$ is given by
\begin{align}
\Pr \left[K_{j,n,0}=k\ |\ j_0=j\right]
\approx\sum_{i\in \mathcal I_{j,n}}\frac{p_{j,i}}{T_{j,n}}\sum_{ \mathcal X\in  \left\{\mathcal S \subseteq \mathcal N_{j,i,-n}: |\mathcal S|=k-1\right\} }\prod\limits_{m\in \mathcal X}\left(1-b_{j,m}\right)\prod\limits_{m\in {\mathcal N_{j,i,-n}\setminus \mathcal X}}b_{j,m},\label{eqn:K-pmf}
\end{align}
where $k=1,\cdots, K_{j}$, and\footnote{Note that $\widehat{A}_{j,m}\left(T_{j,m},T_{\overline j,m}\right) = \frac{A_{j,m}\left(T_{j,m},T_{\overline j,m}\right)}{T_{j,m}}$.}
\begin{align}
&b_{j,m}\triangleq\left(1+\frac{a_m\lambda_u\widehat{A}_{j,m}\left(T_{j,m},T_{\overline j,m}\right)}{3.5\lambda_{j}}\right)^{-3.5}\label{eqn:file-requ-prob},\\
&\widehat{A}_{j,m}(T_{j,m},T_{\overline j,m})\triangleq\frac{\lambda_j}{\lambda_jT_{j,m} + \lambda_{\overline j}T_{\overline j,m}\left(\frac{P_{\overline j}}{P_j}\right)^{ \frac{2}{\alpha} }}. 
\end{align}
\label{Lem:pmf-K}
\end{Lem}
\begin{proof}
Please refer to Appendix A.
\end{proof}

\begin{figure*}[!t]
\footnotesize{\begin{align}
f_{j,k}(x,y)\triangleq &\, 2\pi\lambda_j\int_{0}^{\infty} d\exp\left(-\pi\lambda_j\left (\theta_{1,k}x+\theta_{2,j,k}y+\theta_{3,j,k}\right )d^2\right)\exp\left(- \left(2^{\frac{k\tau}{W}}-1\right) d^{\alpha}\frac{N_0}{P_j}\right){\rm d}d.\label{eqn:f-j-k}\\
\theta_{1,k}=&\frac{2}{\alpha}\left(2^{\frac{k\tau}{W}}-1\right)^{\frac{2}{\alpha}} \left(B'\left(\frac{2}{\alpha},1-\frac{2}{\alpha}, 2^{-\frac{k\tau}{W}} \right)-B\left(\frac{2}{\alpha},1-\frac{2}{\alpha}\right)\right) + 1.\label{eqn:c_1_k} \\
\theta_{2,j,k}=&\frac{2\lambda_{\overline j}}{\alpha\lambda_j}\left(\sigma_{\overline j}\left(2^{\frac{k\tau}{W}}-1\right) \right)^{\frac{2}{\alpha}}\left(B'\left(\frac{2}{\alpha},1-\frac{2}{\alpha}, 2^{-\frac{k\tau}{W}} \right)-B\left(\frac{2}{\alpha},1-\frac{2}{\alpha}\right)\right)+\frac{\lambda_{\overline j}}{\lambda_j}\sigma_{\overline j}^{\frac{2}{\alpha}}.\label{eqn:c_2_k}\\
\theta_{3,j,k}=&\frac{2}{\alpha}\left(2^{\frac{k\tau}{W}}-1\right) ^{\frac{2}{\alpha}}B\left(\frac{2}{\alpha},1-\frac{2}{\alpha}\right) +\frac{2\lambda_{\overline j}}{\alpha\lambda_j}\left(\sigma_{\overline j}\left(2^{\frac{k\tau}{W}}-1\right) \right)^{\frac{2}{\alpha}}B\left(\frac{2}{\alpha},1-\frac{2}{\alpha}\right).\label{eqn:c_3_k}
\end{align}}
\normalsize \hrulefill
\end{figure*}

\begin{Thm} [Performance]
The successful transmission probability  $q\left(\mathbf p_{1},\mathbf p_{2}\right)$ of $u_{0}$ is
\begin{align}
q\left(\mathbf p_{1},\mathbf p_{2}\right)=q_1\left(\mathbf p_{1},\mathbf p_{2}\right)+q_2\left(\mathbf p_{2},\mathbf p_{1}\right)\label{eqn:CPrate_multifile_noise},
\end{align}
where
\begin{align}
&q_j(\mathbf p_j,\mathbf p_{\overline j}) \nonumber\\
=&\sum_{n\in \mathcal N}a_{n} \sum_{k=1}^{K_j}\left( \sum_{i\in \mathcal I_{j,n}}p_{j,i}\sum_{ \mathcal X\in  \left\{\mathcal S \subseteq \mathcal N_{j,i,-n}: |\mathcal S|=k-1\right\} }\prod\limits_{m\in \mathcal X}\left(1-b_{j,m}\right)\prod\limits_{m\in {\mathcal N_{j,i,-n}\setminus \mathcal X}}b_{j,m} \right) f_{j,k}(T_{j,n},T_{\overline j,n}), 
\end{align}
$b_{j,m}$ is given by \eqref{eqn:file-requ-prob} and $f_{j,k}(T_{j,n},T_{\overline j,n})$ is given by \eqref{eqn:f-j-k} with $ \theta_{1,k} $, $\theta_{2,j,k} $ and $ \theta_{3,j,k} $ given by~\eqref{eqn:c_1_k}, \eqref{eqn:c_2_k} and \eqref{eqn:c_3_k}. Here, $B^{'}\left(x,y,z\right)\triangleq \int_{z}^{1}u^{x-1}\left(1-u\right)^{y-1}{\rm d}u$ and $B(x,y)\triangleq\int_{0}^{1}u^{x-1}\left(1-u\right)^{y-1}{\rm d}u$ denote the complementary incomplete Beta function and  the Beta function, respectively. \label{Thm:generalKmulti}
\end{Thm}

From Theorem~\ref{Thm:generalKmulti}, we can see that in the general region, the physical layer parameters $\alpha$, $W$, $\lambda_1$, $\lambda_2$, $\lambda_u$, $\frac{P_1}{N_0}$, $\frac{P_2}{N_0}$ and the design parameters $\left(\mathbf p_1,\mathbf p_2\right)$
jointly affect the successful transmission probability $q\left(\mathbf p_1,\mathbf p_2\right)$.
The impacts of the physical layer parameters and the design parameters on $q\left(\mathbf p_1,\mathbf p_2\right)$  are coupled in a complex manner.

\subsection{Performance Analysis in Asymptotic Region}

The gain of multicasting over unicasting increases with user density\cite{cui2017analysis}. In this subsection, to obtain design insights into caching and multicasting, we analyze the asymptotic successful transmission probability  in the high SNR and high user density region.
Note that in the rest of the paper, when considering the asymptotic region (i.e., the high SNR and user density region), we assume
$\frac{P_1}{N_0}\to \infty$ and $ \frac{P_2}{N_0}\to\infty $ while fixing the power ratio, i.e., $ \sigma_1 $ ($ \sigma_2 $).
In addition, in the high user density region where $\lambda_u\to \infty $, the discrete random variable $K_{j,n,0}\to K_j$ in distribution.
From Theorem~\ref{Thm:generalKmulti}, we can derive the successful transmission probability in the asymptotic region.

\begin{Cor}[Asymptotic Performance]
When $\frac{P}{N_0}\to \infty$ and $\lambda_u\to \infty$,
\begin{align}
q(\mathbf p_1,\mathbf p_2) = q_{1, \infty}(\mathbf T_1,\mathbf T_2 ) + q_{2, \infty}(\mathbf T_2,\mathbf T_1) \triangleq q(\mathbf T_1,\mathbf T_2), \label{eqn:f-k-infty-sym}
\end{align}
where
\begin{align}
q_{j, \infty}(\mathbf T_j,\mathbf T_{\overline j})=\sum_{n\in\mathcal N} \frac{a_nT_{j,n}}{ \theta_{1,K_j}T_{j,n}+\theta_{2,j,K_j}T_{\overline j,n}+\theta_{3,j,K_j} }. \label{eqn:f-j-k-infty-sym}
\end{align}
Here, $T_{j,n}$ is given by \eqref{eqn:def-T-n}, and $\theta_{1,k}$, $\theta_{2,j,k}$ and $\theta_{3,j,k}$ are given by \eqref{eqn:c_1_k}, \eqref{eqn:c_2_k} and \eqref{eqn:c_3_k}.
\label{Cor:asym-perf-v2}
\end{Cor}
\begin{proof}
Please refer to Appendix B.
\end{proof}

\begin{figure}[t]
\begin{center}
\includegraphics[width=12cm]{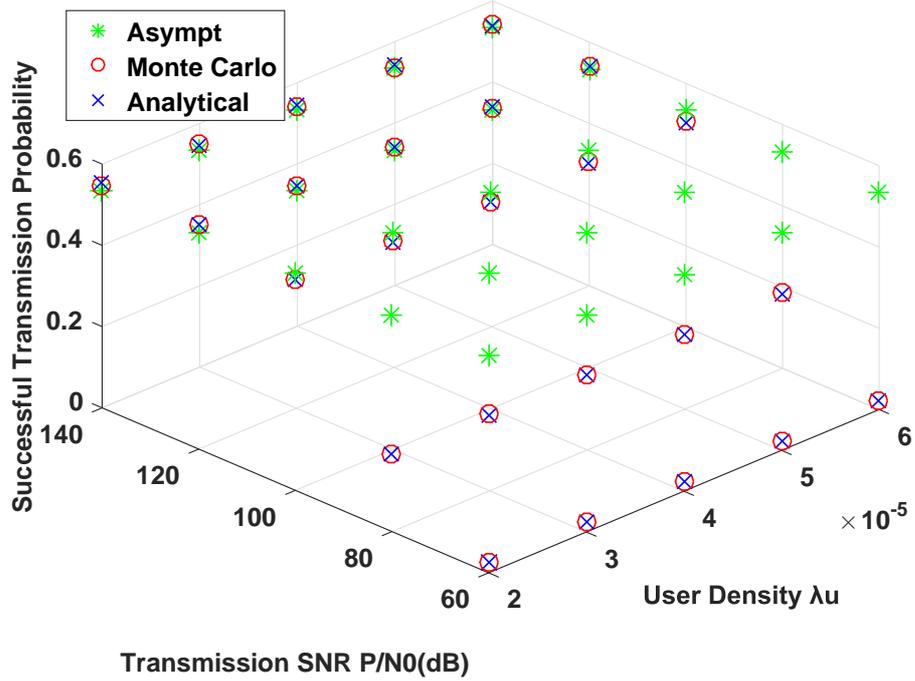}
\end{center}
\caption{\small{Successful transmission probability versus SNR $\frac{P}{N_0}$  and user density $\lambda_u$.  $N=10$, $K_1=3$, $K_2=2$, $p_{1,i} = \frac{1}{\binom{10}{3}}$ for all $i=1,2,\cdots,\binom{10}{3}$, $p_{2,i} = \frac{1}{\binom{10}{2}}$ for all $i=1,2,\cdots,\binom{10}{2}$, $\lambda_{1}=5\times10^{-7}$, $\lambda_{2}=3\times10^{-6}$, $P_1=10^{1.5}P$, $P_2=P$, $\alpha=4$, $W = 20\times 10^6$, $\tau =  35\times10^4$ and $a_n=\frac{n^{-\gamma}}{\sum_{n\in \mathcal N}n^{-\gamma}}$ with $\gamma =1$.}}
\label{fig:verification-Kmulti}
\end{figure}
Note that $ q_{j,\infty}(\mathbf T_j,\mathbf T_{\overline j}) = \lim_{\frac{P}{N_0}\to \infty , \lambda_u\to \infty}q_{j}(\mathbf p_j,\mathbf p_{\overline j})$ and $q_{\infty}(\mathbf T_1,\mathbf T_2) = \lim_{\frac{P}{N_0}\to \infty , \lambda_u\to \infty} q(\mathbf p_1,\mathbf p_2)$; when $\lambda_u \to \infty$ (corresponding to the full file load case), $q_j$ and $q$ become functions of $\mathbf T_1$ and $\mathbf T_2$ instead of $\mathbf p_1$ and $\mathbf p_2$. In addition, the asymptotic successful transmission probability in Corollary~\ref{Cor:asym-perf-v2} and the performance metric in\cite{li2016optimization,wen2016cache} have different meanings, although they share similar forms.
From Corollary~\ref{Cor:asym-perf-v2}, we can see that in the high SNR and high user density region, the impact of the physical layer parameters $\alpha$, $W$, $\lambda_j$ and $\sigma_j$, captured by $\theta_{1,j}$, $\theta_{2,j,K_j}$ and $\theta_{3,j,K_j}$, and the impact of the design parameters $\left(\mathbf p_1,\mathbf p_2\right)$ on $q_{\infty}\left(\mathbf T_1,\mathbf T_2\right)$ can be easily separated. In most practical cases, $ \theta_{1,K_1},\theta_{1,K_2}>0 $. Thus, we consider $ \theta_{1,K_1},\theta_{1,K_2}>0 $ in the rest of the paper.

Fig.~\ref{fig:verification-Kmulti}  verifies Theorem~\ref{Thm:generalKmulti} and Corollary~\ref{Cor:asym-perf-v2}, and demonstrates the accuracy of  the approximation adopted.
Fig.~\ref{fig:verification-Kmulti} also indicates that  $q_{\infty}\left(\mathbf T_1, \mathbf T_2\right)$  provides a simple and good approximation for   $q\left(\mathbf p_1, \mathbf p_2\right)$  in  the high SNR (e.g., $\frac{P}{N_0}\geq 120$ dB) and the high user density region (e.g., $\lambda_u\geq 3\times 10^{-5}$).

In the asymptotic region, from \cite{cui2017analysis}, we know that the constraints on $ \left(\mathbf p_1,\mathbf p_2\right) $ in \eqref{eqn:cache-constr-indiv} and \eqref{eqn:def-T-n} can be equivalently  rewritten as $ \left(\mathbf T_1,\mathbf T_2\right)\in \mathcal T_1\times\mathcal T_2 $, where $ \mathcal T_j $ is defined as
\begin{align}
\mathcal T_j \triangleq\left\{\mathbf T_j \ \middle|\ 0\leq T_{j,n}\leq 1,n\in\mathcal N, \sum\limits_{n\in\mathcal N}T_{j,n}=K_j \right\}. \label{eqn:strategy-set-j}
\end{align} 
To obtain design insights into caching in large-scale multi-tier wireless multicasting networks, in Section~\ref{Sec:joint-design} and Section~\ref{Sec:compet-design}, we focus on the joint and competitive caching designs in the asymptotic region, respectively.

\section{Joint Caching Design}\label{Sec:joint-design}

In this section, we consider the case that the two tiers of POAs are managed by a single operator, e.g., as in a HetNet. We first formulate the optimal joint caching design problem to maximize the successful transmission probability in the asymptotic region. Then, we develop an algorithm to obtain a stationary point.

\subsection{Optimization Problem Formulation}\label{Sec:opt-prob-formulate}

In this subsection, we formulate the optimal joint caching design problem to maximize the successful transmission probability $ q_{\infty}\left(\mathbf T_1,\mathbf T_2\right) $ by optimizing  the caching distributions  of the two tiers, i.e., $ \left(\mathbf T_1,\mathbf T_2\right) $.
\begin{Prob}[Joint Caching Design]\label{prob:opt-asymp-T}
	\begin{align}
	q_{\infty}^* \triangleq \max_{\mathbf{T}_1, \mathbf{T}_2} &\quad  q_{\infty}\left(\mathbf{T}_1, \mathbf{T}_2\right)\nonumber\\
	\text{s.t.} & \quad \mathbf T_j\in\mathcal T_j, \nonumber
	\end{align}
	where $ q_{\infty}\left(\mathbf{T}_1, \mathbf{T}_2\right) $ is given by \eqref{eqn:f-k-infty-sym} and $ \mathcal T_j $ is given by~\eqref{eqn:strategy-set-j}.
\end{Prob}

Problem~\ref{prob:opt-asymp-T} maximizes a differentiable (nonconcave in general) function over a convex set, and it is thus nonconvex in general. Note that Problem~\ref{prob:opt-asymp-T} and Problem 0 in \cite{wen2016cache} are mathematically equivalent, although this paper and \cite{wen2016cache} have different scopes. In the following subsection, we propose an efficient algorithm to solve Problem~\ref{prob:opt-asymp-T}. In contrast, \cite{wen2016cache} simplifies the nonconvex problem to a convex one, and uses the optimal solution to the simplified problem as a sub-optimal solution to the original problem, which may not provide performance guarantee.

\subsection{Algorithm Design}\label{Sec:new-opt}

Recall that Problem~\ref{prob:opt-asymp-T} is to maximize a differentiable (nonconcave in general) function over a convex set.
We can obtain a stationary point of Problem~\ref{prob:opt-asymp-T} using the  gradient projection method with a diminishing stepsize\cite[pp. 227]{bertsekas1999nonlinear}, as summarized in Algorithm~\ref{alg:local-asymp-sym} for completeness. 
In Algorithm~\ref{alg:local-asymp-sym}, the diminishing stepsize $ \epsilon(t) $ satisfies $ \epsilon(t)\to 0 $ as $ t\to\infty $, $ \sum\limits_{t=1}^{\infty} \epsilon(t) = \infty $ and $ \sum\limits_{t=1}^{\infty} \left (\epsilon(t)\right )^2 < \infty $. In addition, Step 3 is the projection of $ \bar T_{j,n}(t+1) $ onto set $ \mathcal T_j $. 
It is shown in \cite[pp. 229]{bertsekas1999nonlinear} that the sequence $ \left \{(\mathbf T_1(t),\mathbf T_2(t))\right \} $ generated by Algorithm~\ref{alg:local-asymp-sym} converges to a stationary point of Problem~\ref{prob:opt-asymp-T}. Note that a stationary point is a point that satisfies the necessary optimality conditions of a nonconvex optimization problem, and it is the classic goal in the design of iterative algorithms for nonconvex optimization problems.
However, the rate of convergence of Algorithm~\ref{alg:local-asymp-sym} is strongly dependent on the choices of stepsize $ \epsilon(t) $. If it is chosen improperly,  it may take a large number of iterations for Algorithm~\ref{alg:local-asymp-sym} to meet some convergence criterion.

\begin{algorithm}[t]
\caption{Stationary Point of Problem \ref{prob:opt-asymp-T} Based on the Standard Gradient Projection Method}
\footnotesize{\begin{algorithmic}[1]
\STATE Initialize  $t=1$ and choose any $ \mathbf T_j(1)\in\mathcal T_j $ (e.g., $T_{j,n}(1)=\frac{K_j}{N}$  for all $n\in \mathcal N$), $ j=1,2 $. 
\STATE   For all $n\in \mathcal N$, compute $\bar T_{j,n}(t+1)$ according to $\bar T_{j,n}(t+1)=T_{j,n}(t)+\epsilon(t)\frac{\partial q_{\infty}\left(\mathbf T_1(t),\mathbf T_2(t)\right)}{\partial T_{j,n}(t)}$.
\STATE For all $n\in \mathcal N$, compute $T_{j,n}(t+1)$ according to $T_{j,n}(t+1)=\min\left\{\left[\bar T_{j,n}(t+1)-\nu_j^*\right]^+,1\right\}$,
where $\nu_j^*$ satisfies $\sum_{n\in \mathcal N}\min\left\{\left[\bar T_{j,n}(t+1)-\nu_j^*\right]^+,1\right\}=K_j$.
\STATE Set $t=t+1$ and go to Step 2.
\end{algorithmic}}\label{alg:local-asymp-sym}
\end{algorithm}

To address the above problem, in this subsection we propose an iterative algorithm to obtain a stationary point of Problem~\ref{prob:opt-asymp-T} more efficiently.
This algorithm is  based on  the block successive upper-bound minimization  algorithm originally proposed in\cite{razaviyayn2013unified}.
It alternatively updates $ \mathbf T_1 $ and $ \mathbf T_2 $ by maximizing an approximate function of $ q_{\infty}(\mathbf T_1,\mathbf T_2) $, which is successively refined so that eventually the iterative algorithm can converge to a stationary point of Problem~\ref{prob:opt-asymp-T}.
Specifically, at iteration $ t $, we update the caching distribution of the $ j $th tier by maximizing the approximate function of $ q_{\infty}\left(\mathbf T_1,\mathbf T_2\right) $ given the caching distribution of the $ \overline j $th tier, and fix the caching distribution of the $ \overline j $th tier, where $ j=((t+1)\mod 2)+1 $.

For notational convenience, we define
\begin{align}
\tilde q_{\infty}\left (\mathbf T_j,\mathbf T_{\overline j}\right )\triangleq
\begin{cases}
q_{\infty}\left (\mathbf T_j,\mathbf T_{\overline j}\right ), \ j=1, \\
q_{\infty}\left (\mathbf T_{\overline j},\mathbf T_j\right ), \ j=2. 
\end{cases}
\end{align}
At iteration $t$,
choose $ g_j(\mathbf T_{j};\mathbf T_1(t),\mathbf T_2(t)) $ to be an approximate function of $ \tilde q_{\infty} (\mathbf T_j,\mathbf T_{\overline j}(t) ) $, where $ g_j(\mathbf T_{j};\mathbf T_1(t),\mathbf T_2(t)) $ is given by
\begin{align}
&g_{j}\left(\mathbf T_{j};\mathbf T_1(t),\mathbf T_2(t)\right) \nonumber\\
\triangleq\,&  q_{j,\infty}(\mathbf T_{j},\mathbf T_{\overline j}(t)) + q_{\overline j,\infty}(\mathbf T_{\overline j}(t),\mathbf T_{j}(t))
+ \sum\limits_{n\in\mathcal N}\frac{\partial q_{\overline j,\infty}\left(\mathbf T_{\overline j}(t),\mathbf T_j(t)\right)}{\partial T_{j,n}}\left(T_{j,n}-T_{j,n}(t)\right) \nonumber\\
=\,&q_{j,\infty}(\mathbf T_{j},\mathbf T_{\overline j}(t)) + q_{\overline j,\infty}(\mathbf T_{\overline j}(t),\mathbf T_{j}(t))
- \sum\limits_{n\in\mathcal N}\frac{a_n\theta_{2,\overline j,K_{\overline j}}T_{\overline j,n}(t) \left(T_{j,n}-T_{j,n}(t)\right)}{ \left(\theta_{1,K_{\overline j}}T_{\overline j,n}(t) + \theta_{2,\overline j,K_{\overline j}}T_{j,n}(t) + \theta_{3,\overline j,K_{\overline j}}\right)^2 }. \label{eqn:appr-g-def}
\end{align}
Note that the first concave component function of $ \tilde q_{\infty} (\mathbf T_j,\mathbf T_{\overline j}(t) ) $, i.e., $ q_{j,\infty}(\mathbf T_{j},\mathbf T_{\overline j}(t)) $ is left unchanged, and only the second nonconcave (actually convex) component function, i.e., $ q_{\overline j,\infty}(\mathbf T_{\overline j}(t),\mathbf T_{j}) $ is linearized at $ \mathbf T_{j}=\mathbf T_{j}(t) $.
This choice of the approximate function is beneficial from several aspects.
Firstly, it can guarantee the convergence of the algorithm to a stationary point of Problem~\ref{prob:opt-asymp-T}, which will be seen in Theorem~\ref{Thm:BSUM-cvg}.
Secondly, the partial concavity of the original objective function is preserved, and the resulting algorithm typically converges much faster than Algorithm~\ref{alg:local-asymp-sym}, where all component functions are linearized and no partial concavity is  exploited.
Thirdly, it yields a closed-form optimal solution to the optimization problem at each iteration, which will be explained in Lemma~\ref{Lem:BSUM-subprob-opt}.
Specifically, $ g_j(\mathbf T_{j};\mathbf T_1(t),\mathbf T_2(t)) $ is strictly concave on $ \mathcal T_j $ for any given $ \left(\mathbf T_1(t),\mathbf T_2(t)\right)\in\mathcal T_1\times\mathcal T_2 $, and satisfies\footnote{Note that \eqref{eqn:BSUM-tight} holds since $ g_j(\mathbf T_{j};\mathbf T_1(t),\mathbf T_2(t)) $ and $ \tilde q_{\infty}(\mathbf T_j,\mathbf T_{\overline j}(t)) $ have the same value at the point where $ g_j(\mathbf T_{j};\mathbf T_1(t),\mathbf T_2(t)) $  is defined, i.e., $ (\mathbf T_1,\mathbf T_2) = (\mathbf T_1(t),\mathbf T_2(t)) $. \eqref{eqn:BSUM-lower-bound} holds since $ q_{\overline j,\infty} (\mathbf T_{\overline j},\mathbf T_j ) $ is a convex function of $ \mathbf T_j $ for any given $ \mathbf T_{\overline j}\in\mathcal T_{\overline j} $.}
\begin{align}
g_j\left(\mathbf T_j(t);\mathbf T_1\left(t\right),\mathbf T_2\left(t\right)\right) 
&= \tilde q_{\infty}\left(\mathbf T_j\left(t\right),\mathbf T_{\overline j}\left(t\right)\right), \ \left(\mathbf T_1(t),\mathbf T_2(t)\right)\in\mathcal T_1\times\mathcal T_2, \label{eqn:BSUM-tight}\\
g_j\left(\mathbf T_j;\mathbf T_1\left(t\right),\mathbf T_2\left(t\right)\right) 
&\leq \tilde q_{\infty}\left(\mathbf T_j,\mathbf T_{\overline j}\left(t\right)\right),\ \mathbf T_j\in\mathcal T_j,\ \left(\mathbf T_1\left(t\right),\mathbf T_2\left(t\right)\right)\in\mathcal T_1\times\mathcal T_2,  \label{eqn:BSUM-lower-bound}
\end{align}
The conditions in \eqref{eqn:BSUM-tight} and \eqref{eqn:BSUM-lower-bound} imply that  $ g_j\left(\mathbf T_j;\mathbf T_1(t),\mathbf T_2(t)\right) $ is a tight lower-bound of $ \tilde q_{\infty}(\mathbf T_j,\mathbf T_{\overline j}(t)) $. 
The differentiability of $ g_j\left(\mathbf T_j;\mathbf T_1(t),\mathbf T_2(t)\right) $ guarantees that the first-order behavior of $ g_j\left(\mathbf T_j;\mathbf T_1(t),\mathbf T_2(t)\right) $ is the same as $ \tilde q_{\infty}(\mathbf T_j,\mathbf T_{\overline j}(t)) $ locally. 
At each iteration $ t $, we update the caching distribution of the $ j $th tier given the caching distribution of the $ \overline j $th tier by solving the following problem, and fix the caching distribution of the $ \overline j $th tier, where $ j=((t+1)\text{ mod }2)+1 $.
\begin{Prob} [Optimization at Iteration $ t $]
\label{prob:BSUM-subprob}
For tier $ j=((t+1)\text{ mod }2)+1 $, we have
\begin{align}
\mathbf T_j(t+1)=\argmax_{\mathbf T_j} &\quad g_j\left(\mathbf T_j;\mathbf T_1(t),\mathbf T_2(t)\right) \nonumber\\
\text{s.t.} &\quad \mathbf T_j\in\mathcal T_j, \nonumber
\end{align}
where $ g_j\left(\mathbf T_j;\mathbf T_1(t),\mathbf T_2(t)\right) $ is given by~\eqref{eqn:appr-g-def}.
\end{Prob}

Problem~\ref{prob:BSUM-subprob} is a convex optimization problem and Slater's condition is satisfied, implying that strong duality holds. Using KKT conditions, we can obtain the closed-form optimal solution to Problem~\ref{prob:BSUM-subprob}, as shown in the following lemma.
\begin{Lem} [Optimal Solution to Problem~\ref{prob:BSUM-subprob}]
For all $ j=((t+1)\text{ mod }2)+1 $, the optimal solution to Problem~\ref{prob:BSUM-subprob}  is given by
\begin{align}
&T_{j,n}(t+1)=\notag\\
&\min{\left\{\left[\frac{1}{\theta_{1,K_{j}}}\sqrt{\frac{a_n\left(\theta_{2,j,K_{j}}T_{\overline j,n}(t)+\theta_{3,j,K_{j}}\right)}{ \nu_j^*(t) + \frac{a_n\theta_{2,{\overline j},K_{\overline j}}T_{\overline j,n}(t)}{\left(\theta_{1,K_{\overline j}}T_{\overline j,n}(t) + \theta_{2,\overline j,K_{\overline j}}T_{j,n}(t) + \theta_{3,\overline j,K_{\overline j}}\right)^2} }} - \frac{\theta_{2,j,K_{j}}T_{\overline j,n}(t)+\theta_{3,j,K_{j}}}{\theta_{1,K_{j}}}\right]^+,1\right\}}, \ n\in\mathcal N, \notag
\end{align}
where $[x]^+\triangleq\max\{x,0\}$ and $\nu_j^*(t)$ is the Lagrange multiplier that satisfies 
\begin{align}
\sum_{n\in\mathcal N}T_{j,n}(t+1) = K_{j}.\notag
\end{align}
\label{Lem:BSUM-subprob-opt}
\end{Lem}

Note that $ \nu_j^*(t) $ can be efficiently obtained by using bisection search.
\begin{algorithm}[t]
\caption{Stationary Point of Problem \ref{prob:opt-asymp-T} Based on BSUM}
\footnotesize{\begin{algorithmic}[1]
\STATE Initialize  $t=1$ and choose any $ \mathbf T_j(1)\in\mathcal T_j $ (e.g., $T_{j,n}(1)=\frac{K_j}{N}$  for all $n\in \mathcal N$), $ j=1,2 $. 
\STATE  Compute $ j=((t+1) \text{ mod } 2)+1 $.
\STATE For all $ n\in\mathcal N $, compute $ T_{j,n}(t+1) $ according to Lemma~\ref{Lem:BSUM-subprob-opt}. 
\STATE For all $ n\in\mathcal N $, set $ T_{{\overline j},n}(t+1)=T_{{\overline j},n}(t) $.
\STATE Set $t=t+1$ and go to Step 2.
\end{algorithmic}}\label{alg:BSUM}
\end{algorithm}
The details of the proposed iterative algorithm are summarized in Algorithm~\ref{alg:BSUM}.
Based on the conditions in \eqref{eqn:BSUM-tight} and \eqref{eqn:BSUM-lower-bound}, 
we show the convergence and optimality of Algorithm~\ref{alg:BSUM}. 
\begin{Thm} [Convergence and Optimality of Algorithm~\ref{alg:BSUM}]
\label{Thm:BSUM-cvg}
The sequence $ \left \{q_{\infty}(\mathbf T_1(t),\mathbf T_2(t))\right \} $ generated by Algorithm~\ref{alg:BSUM} is convergent, and every limit point of $ \left \{(\mathbf T_1(t),\mathbf T_2(t))\right \} $ is a stationary point of Problem~\ref{prob:opt-asymp-T}.
\end{Thm}
\begin{proof}
Please refer to Appendix C. 
\end{proof}

Different from Algorithm~\ref{alg:local-asymp-sym}, Algorithm~\ref{alg:BSUM} does not rely on a stepsize. Thus, Algorithm~\ref{alg:BSUM} may have more robust convergence performance than Algorithm~\ref{alg:local-asymp-sym}, as we shall illustrate later in Fig.~\ref{fig:simulation-large-convergence}.

In the rest of this subsection, we
consider a special case where $ K_1=K_2\triangleq K $. In this case, $ \lambda_1 P_1^{\frac{2}{\alpha}}\theta_{1,K_1} = \lambda_2 P_2^{\frac{2}{\alpha}}\theta_{2,2,K_2}\triangleq\mu_{1,K} $, $ \lambda_1 P_1^{\frac{2}{\alpha}}\theta_{2,1,K_1} = \lambda_2 P_2^{\frac{2}{\alpha}}\theta_{1,K_2}\triangleq\mu_{2,K} $ and $ \lambda_1 P_1^{\frac{2}{\alpha}}\theta_{3,1,K_1} = \lambda_2 P_2^{\frac{2}{\alpha}}\theta_{3,2,K_2}\triangleq\mu_{3,K} $.
In addition, $ q_{\infty}(\mathbf T_1,\mathbf T_2) $ can be further simplified as
\begin{align}
q_{\infty}(\mathbf T_1,\mathbf T_2) = \sum\limits_{n\in\mathcal N}a_n\frac{\lambda_1 P_1^{\frac{2}{\alpha}} T_{1,n} + \lambda_2 P_2^{\frac{2}{\alpha}} T_{2,n}}{\mu_{1,K}T_{1,n} + \mu_{2,K}T_{2,n} + \mu_{3,K}}, 
\end{align}
which is a concave function of $ (\mathbf T_1,\mathbf T_2) $. 
Thus, Problem~\ref{prob:opt-asymp-T} becomes a convex optimization problem, and a (globally) optimal solution can be obtained by standard convex optimization methods such as interior-point methods. 
However, when $ N $ is very large, standard convex optimization methods may not scale very well.
Motivated by\cite{wen2016cache}, by exploring structural properties of Problem~\ref{prob:opt-asymp-T} in this case, we develop a low-complexity algorithm to obtain an optimal solution. The method consists of two stages. 
In the first stage, we solve a relaxed version of Problem~\ref{prob:opt-asymp-T} to obtain a system of linear equations of an optimal solution to Problem~\ref{prob:opt-asymp-T}. This stage is the same as that in\cite{wen2016cache}, and is included for completeness. 
Denote $ \mathbf R \triangleq (R_n)_{n\in\mathcal N} $, where $ R_n \triangleq P_1^{\frac{2}{\alpha}}\lambda_1T_{1,n} + P_2^{\frac{2}{\alpha}}\lambda_2T_{2,n} $. Specifically, Problem~\ref{prob:opt-asymp-T} can be relaxed as follows. 
\begin{Prob} [Relaxed Version of Problem~\ref{prob:opt-asymp-T} When $ K_1=K_2=K $~\cite{wen2016cache}]
\label{prob:relax-social-opt}
\begin{align}
\max_{\mathbf R}&\quad\sum\limits_{n\in\mathcal N}a_n\frac{R_n}{\theta_{1,K}R_n + \mu_{3,K}},\nonumber\\
s.t. &\quad 0\leq R_n\leq P_1^{\frac{2}{\alpha}}\lambda_1 + P_2^{\frac{2}{\alpha}}\lambda_2, \ n\in\mathcal N,\nonumber\\
&\quad \sum\limits_{n\in\mathcal N}R_n = (P_1^{\frac{2}{\alpha}}\lambda_1 + P_2^{\frac{2}{\alpha}}\lambda_2)K. \nonumber
\end{align}
Let $ \mathbf R^* $ denote the optimal solution to Problem~\ref{prob:relax-social-opt}. 
\end{Prob}

The optimal soluton to Problem~\ref{prob:relax-social-opt} is given by \cite[Proposition 3]{wen2016cache}, i.e.,
\begin{align}
R_n^* = \min\left \{\left [\frac{1}{\theta_{1,K}}\left (\sqrt{\frac{a_n\mu_{3,K}}{\nu^*}} - \mu_{3,K}\right )\right ]^+, P_1^{\frac{2}{\alpha}}\lambda_1 + P_2^{\frac{2}{\alpha}}\lambda_2\right \}, \ n\in\mathcal N, \label{eqn:opt-R}
\end{align}
where $ \nu^* $ is the Lagrange multiplier that satisfies
\begin{align}
\sum\limits_{n\in\mathcal N}R_n^* = \left (P_1^{\frac{2}{\alpha}}\lambda_1 + P_2^{\frac{2}{\alpha}}\lambda_2\right )K. \label{eqn:opt-R-nu}
\end{align}
Note that $ \nu^* $ can be efficiently obtained by using bisection search. In addition, by Proposition 4 in \cite{wen2016cache}, we know that the optimal solution to Problem~\ref{prob:relax-social-opt} and an optimal solution to Problem~\ref{prob:opt-asymp-T} satisfy a system of linear equations: 
\begin{align}
P_1^{\frac{2}{\alpha}}\lambda_1T_{1,n}^* + P_2^{\frac{2}{\alpha}}\lambda_2 T_{2,n}^* = R_n^*, \ n\in\mathcal N. \label{eqn:R-to-T-opt}
\end{align}
In the second stage, we solve the system of linear equations given in~\eqref{eqn:R-to-T-opt} to obtain an optimal solution $ \left (\mathbf T_1^*,\mathbf T_2^*\right ) $ to Problem~\ref{prob:opt-asymp-T}. 
In our case, we can easily show that 
\begin{align}
T_{j,n}^*=\frac{R_n^*}{P_1^{\frac{2}{\alpha}}\lambda_1 + P_2^{\frac{2}{\alpha}}\lambda_2}, \ n\in\mathcal N, \ j=1,2 \label{eqn:R-to-T-mapping}
\end{align}
is a solution to the system of linear equations in~\eqref{eqn:R-to-T-opt}. This stage is different from that in\cite{wen2016cache}, as we can directly obtain $ T_{j,n}^{*} $ using the closed-form expression in~\eqref{eqn:R-to-T-mapping}, due to $ K_1=K_2 $. The details are summarized in Algorithm~\ref{alg:seq-opt}. 
Note that the complexity of Algorithm~\ref{alg:seq-opt} is close to that of one iteration of Algorithm~\ref{alg:BSUM}. Thus, the complexity of Algorithm~\ref{alg:seq-opt} is much lower than Algorithm~\ref{alg:BSUM}.

\begin{algorithm}[t]
\caption{Globally Optimal Solution}
\footnotesize{\begin{algorithmic}[1]
\STATE	Obtain $ R_n^* $ by~\eqref{eqn:opt-R} and \eqref{eqn:opt-R-nu}. 
\STATE	Compute $ T_{j,n}^* $ by~\eqref{eqn:R-to-T-mapping}, $ n\in\mathcal N $, $ j=1,2 $. 
\end{algorithmic}}\label{alg:seq-opt}
\end{algorithm}

\section{Competitive Caching Design}\label{Sec:compet-design}

In this section, we study the scenario that the two tiers of POAs are managed by two different operators, e.g., IEEE 802.11 APs of two owners. The two different operators have their own interests and thus cannot be jointly managed. Besides, one operator may be sacrificed in order to achieve the maximum total utility. 
Therefore, we propose a game theoretic approach and adopt a NE as a desirable outcome.
We first formulate the competitive caching design for the two different operators within the framework of game theory.
Then, we characterize a NE of the game and develop an algorithm to obtain a NE.

\subsection{Game Formulation}

In this subsection, we formulate the competitive caching design for the two different operators within the framework of game theory. We consider a strategic noncooperative game, where the two operators are
the players.  The utility function of player $ j $ is the successful transmission probability  for  tier $ j $, i.e., $ q_{j,\infty}(\mathbf T_j,\mathbf T_{\overline j}) $. Each tier $ j $ competes against the other tier $ \overline j $ by
choosing its caching distribution $ \mathbf T_j $ (i.e., strategy or action) in the set of admissible strategies $ \mathcal T_j $
to maximize its utility function, i.e., $ q_{j,\infty}(\mathbf T_j,\mathbf T_{\overline j}) $.

\begin{Prob} [Competitive Caching Game] \label{prob:opt-asymp-game}
For all $ j=1,2 $, we have
\begin{align}
\max_{\mathbf T_j} &\quad q_{j,\infty}(\mathbf T_j,\mathbf T_{\overline j})\nonumber\\
\text{s.t.} &\quad \mathbf T_j \in\mathcal T_{j}, \nonumber
\end{align}
\end{Prob}
where $q_{j,\infty}(\mathbf T_j,\mathbf T_{\overline j})$ is given by \eqref{eqn:f-j-k-infty-sym}  and  $ \mathcal T_{j} $ is given by~\eqref{eqn:strategy-set-j}. Let $\mathcal G $ denote the game.

A solution, i.e., a NE,\footnote{A NE is reached when each player, given the strategy profiles of the others, does not get any performance increase by unilaterally changing his own strategy\cite{scutari2008competitive}.}
of game $ \mathcal G $ is defined as follows.
\begin{Def} [Nash Equilibrium of Game $ \mathcal G $]
\label{def:nash-equilibrium}
A (pure) strategy profile $ (\mathbf T_1^{\dagger},\mathbf T_2^{\dagger})\in \mathcal T_{1}\times \mathcal T_{2} $ is a NE of game $ \mathcal G $ if
\begin{align}
q_{j,\infty}(\mathbf T_j^{\dagger},\mathbf T_{\overline j}^{\dagger}) \geq q_{j,\infty}(\mathbf T_j,\mathbf T_{\overline j}^{\dagger}),\ \mathbf T_j\in \mathcal T_{j},\ j=1,2.
\end{align}
\end{Def}

By Definition~\ref{def:nash-equilibrium}, we know that a NE of game $ \mathcal G $ is given by the following problem.
\begin{Prob}[NE of Game $ \mathcal G $]
For all $ j=1,2 $, we have
\begin{align}
\mathbf T_j^{\dagger} = \argmax\limits_{\mathbf T_j} &\quad q_{j,\infty} (\mathbf T_j,\mathbf T_{\overline j}^{\dagger}), \nonumber\\
\text{s.t.} &\quad \mathbf T_j\in\mathcal T_j. \nonumber
\end{align}
\label{prob:nash-equilibrium}
\end{Prob}

\subsection{Nash Equilibrium}

In this subsection, we characterize a NE of game $ \mathcal G $. 
First, we show  the existence and uniqueness of the NE of game $ \mathcal G $.
\begin{Lem} [Existence and Uniqueness of the NE of Game $ \mathcal G $]
\label{Lem:NE-exis-uniq}
There exists a unique NE of game $ \mathcal G $. 
\end{Lem}
\begin{proof}
Please refer to Appendix D. 
\end{proof}

We now obtain the closed-form expression of the unique  NE of game $ \mathcal G $. Since $ q_{j,\infty}(\mathbf T_j,\mathbf T_{\overline j}) $ is strictly concave on $ \mathcal T_{j} $ for any given $ \mathbf T_{\overline j}\in\mathcal T_{\overline j} $, Problem~\ref{prob:nash-equilibrium} is convex and Slater's condition is satisfied, implying that strong duality holds. Using KKT conditions, we can solve Problem~\ref{prob:nash-equilibrium} and show Lemma~\ref{Lem:NE-constr}.
\begin{Lem} [NE of Game $ \mathcal G $]
\label{Lem:NE-constr}
Game $ \mathcal G $ has a unique NE   $ (\mathbf T_1^{\dagger},\mathbf T_2^{\dagger}) $ which is given by 
\begin{align}
&T_{j,n}^{\dagger}= \nonumber\\ &\min{\left\{\left[\frac{1}{\theta_{1,K_{j}}}\sqrt{\frac{a_n\left(\theta_{2,j,K_{j}}T_{\overline j,n}^{\dagger}+\theta_{3,j,K_{j}}\right)}{\nu_j^{\dagger}}} - \frac{\theta_{2,j,K_{j}}T_{\overline j,n}^{\dagger}+\theta_{3,j,K_{j}}}{\theta_{1,K_{j}}}\right]^+,1\right\}},\  n\in\mathcal N,\ j=1,2, \nonumber
\end{align}
where for all $ j=1,2 $, $\nu_{j}^{\dagger}$ is the Lagrange multiplier that satisties 
\begin{align}
\sum_{n\in\mathcal N}T_{j,n}^{\dagger} = K_{j}. \nonumber
\end{align}
\end{Lem}
\begin{Rem}
The file popularity distribution $ \mathbf a $ and the physical layer parameters (captured in $\theta_{1,K_j}$, $\theta_{2,j,K_j}$ and $\theta_{3,j,K_j}$) jointly affect $ \nu_j^{\dagger} $. Given $ \nu_j^{\dagger} $ and $ T_{\overline j,n}^{\dagger} $, $ n\in\mathcal N $, the physical layer parameters (captured in $\theta_{1,K_j}$, $\theta_{2,j,K_j}$ and $\theta_{3,j,K_j}$) affect the caching probabilities of all the files in the same way, while the popularity of file $ n $ (i.e., $ a_n $) only affects the caching probability of file $ n $ (i.e., $ T_{j,n}^{\dagger} $)~\cite{cui2016analysis}.
\end{Rem}

\subsection{Algorithm Design}
In this subsection, we develop an iterative algorithm to obtain the NE of game $ \mathcal G $. It alternatively updates $ \mathbf T_1 $ while $ \mathbf T_2 $ is fixed and $ \mathbf T_2 $ while $ \mathbf T_1 $ is fixed, by solving the following problem at each iteration $ t $.
\begin{Prob} [Optimization at Iteration $ t $] \label{prob:NE-alg-subprob}
For player $ j = ((t+1)\text{ mod }2)+1$, we have
\begin{align}
\mathbf T_j(t+1) = \argmax_{\mathbf T_j}& \quad q_{j,\infty}(\mathbf T_j,\mathbf T_{\overline j}(t))  \nonumber\\
\text{s.t.} & \quad \mathbf T_j\in\mathcal T_j. \nonumber
\end{align}
\end{Prob}

Similar to Problem~\ref{prob:nash-equilibrium}, using KKT conditions, we can obtain the closed-form expression of the unique optimal solution to Problem~\ref{prob:NE-alg-subprob}. We present it below for completeness.
\begin{Lem} [The Optimal Solution to Problem~\ref{prob:NE-alg-subprob}] \label{Lem:NE-alg-subprob-opt}
For all $ j=((t+1)\text{ mod }2)+1 $, the optimal solution to Problem~\ref{prob:NE-alg-subprob} is given by
\begin{align}
&T_{j,n}(t+1)=\nonumber\\
&\min{\left\{\left[\frac{1}{\theta_{1,K_{j}}}\sqrt{\frac{a_n\left(\theta_{2,j,K_{j}}T_{\overline j,n}(t)+\theta_{3,j,K_{j}}\right)}{\nu_j^{\dagger}(t)}} - \frac{\theta_{2,j,K_{j}}T_{\overline j,n}(t)+\theta_{3,j,K_{j}}}{\theta_{1,K_{j}}}\right]^+,1\right\}},\  n\in\mathcal N, \nonumber
\end{align}
where $\nu_j^{\dagger}(t)$ is the Lagrange multiplier that satisfies 
\begin{align}
\sum_{n\in\mathcal N}T_{j,n}(t+1) = K_{j}. \nonumber
\end{align}
\end{Lem}

Note that Problem~\ref{prob:nash-equilibrium} and Problem~\ref{prob:NE-alg-subprob} in Lemma~\ref{Lem:NE-alg-subprob-opt} share  similar forms. Thus, the NE of game $ \mathcal G $ in Lemma~\ref{Lem:NE-constr} and the solution to Problem~\ref{prob:NE-alg-subprob} share similar forms.
Based on the optimal solution to Problem~\ref{prob:NE-alg-subprob}, at iteration $ t $, we update the strategy of player $ j $, and fix the strategy of player $ \overline j $, where $ j=((t+1)\text{ mod }2)+1 $. The details for obtaining the NE of game $ \mathcal G $ is summarized in Algorithm~\ref{alg:game-NE}.

\begin{algorithm}[t]
	\caption{Nash Equilibrium of Game $ \mathcal G $}
	\footnotesize{\begin{algorithmic}[1]
			\STATE Initialize  $t=1$ and choose any $ \mathbf T_j(1) \in\mathcal T_j $ (e.g.,  $T_{j,n}(1)=\frac{K_j}{N}$  for all $n\in \mathcal N$), $j=1,2$.
			\STATE Compute $ j=((t+1)\text{ mod }2)+1 $.
			\STATE   Compute $ \mathbf T_j(t+1) = \argmax\limits_{\mathbf T_j\in\mathcal T_j}q_{j,\infty}(\mathbf T_j,\mathbf T_{\overline j}(t)) $.
			\STATE Set $ \mathbf T_{\overline j}(t+1) = \mathbf T_{\overline j}(t) $.
			\STATE   Set $t=t+1$ and go to Step 2.
	\end{algorithmic}}\label{alg:game-NE}
\end{algorithm}

Note that in general, it is quite difficult to guarantee that an iterative algorithm can converge to the NE of a game, especially for a large-scale wireless network.   By carefully analyze structural properties of the competitive caching design game, we provide a convergence condition for Algorithm~\ref{alg:game-NE}.

\begin{Thm} [Convergence of Algorithm~\ref{alg:game-NE}]
\label{Thm:NE-alg-cvg}
If
\begin{align}
\max\left\{1,\ \Bigg{\lvert}{1-\frac{\theta_{1,K_1}}{\theta_{3,1,K_1}}}\Bigg{\rvert}\right\}
\max\left\{1,\ \Bigg{\lvert}{1-\frac{\theta_{1,K_{2}}}{\theta_{3,2,K_{2}}}}\Bigg{\rvert}\right\} < 4, 
\end{align}
where $\theta_{1,k}$, $\theta_{2,j,k}$ and $\theta_{3,j,k}$ are given by \eqref{eqn:c_1_k}, \eqref{eqn:c_2_k} and \eqref{eqn:c_3_k},
Algorithm~\ref{alg:game-NE} converges to the unique NE of game $ \mathcal G $ for all $ \mathbf T_j(1)\in\mathcal T_j $, $ j=1,2 $, i.e.,   $ (\mathbf T_1(t),\mathbf T_2(t)) \to (\mathbf T_1^{\dagger},\mathbf T_2^{\dagger}) $ as $ t\to\infty $, where $ (\mathbf T_1^{\dagger},\mathbf T_2^{\dagger}) $ is  given by Lemma~\ref{Lem:NE-constr}.
\end{Thm}
\begin{proof}
	Please refer to Appendix E.
\end{proof}

Note that the convergence condition given in Theorem~\ref{Thm:NE-alg-cvg} can be easily satisfied in most cases we are interested in, which will be shown in Fig.~\ref{fig:simulation-large-convergence}.

\section{Numerical Results}\label{Sec:simu}

In this section, we first illustrate the convergence and complexity of the proposed algorithms. Then, we compare the successful transmission probabilities and caching probabilities of the proposed algorithms with those of existing solutions.
In the simulation, we choose $W = 20\times 10^6$, $\tau =  4\times10^4$, $N = 500$, $\alpha=4$, $\lambda_{1}=5\times10^{-7}$, $\lambda_{2}=3\times10^{-6}$ and $P_1 = 10^{1.6}P_2$. We assume that the popularity follows Zipf distribution, i.e., $a_n=\frac{n^{-\gamma}}{\sum_{n\in \mathcal N}n^{-\gamma}}$, where $\gamma$ is the Zipf exponent. 

\begin{figure}[t]
\begin{center}
\subfigure[\small{$K_1=55$ and $K_2=35$.}]
{\resizebox{7.5cm}{!}{\includegraphics{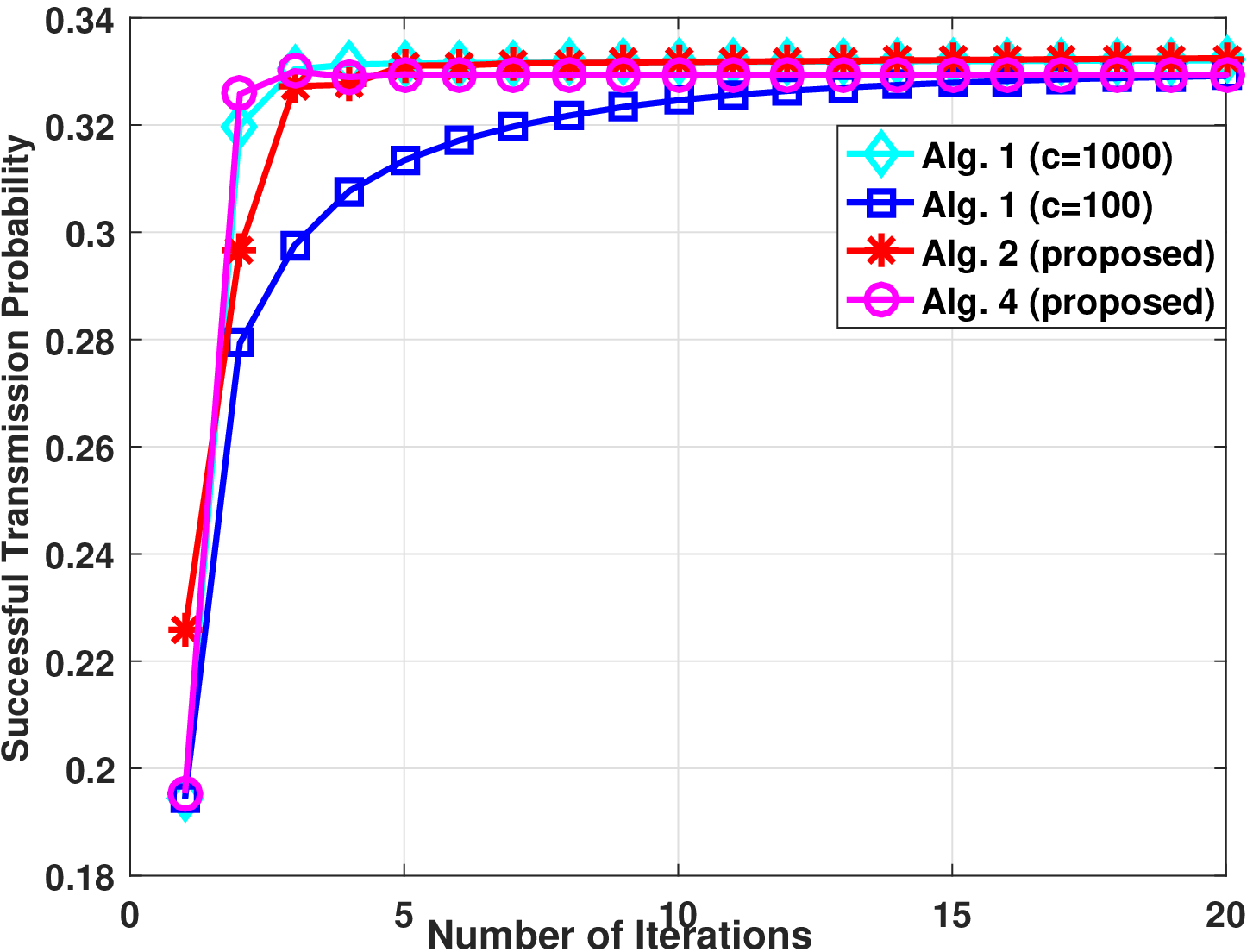}}}
\hspace{4mm}
\subfigure[\small{$K_1=K_2=35$.}]
{\resizebox{7.5cm}{!}{\includegraphics{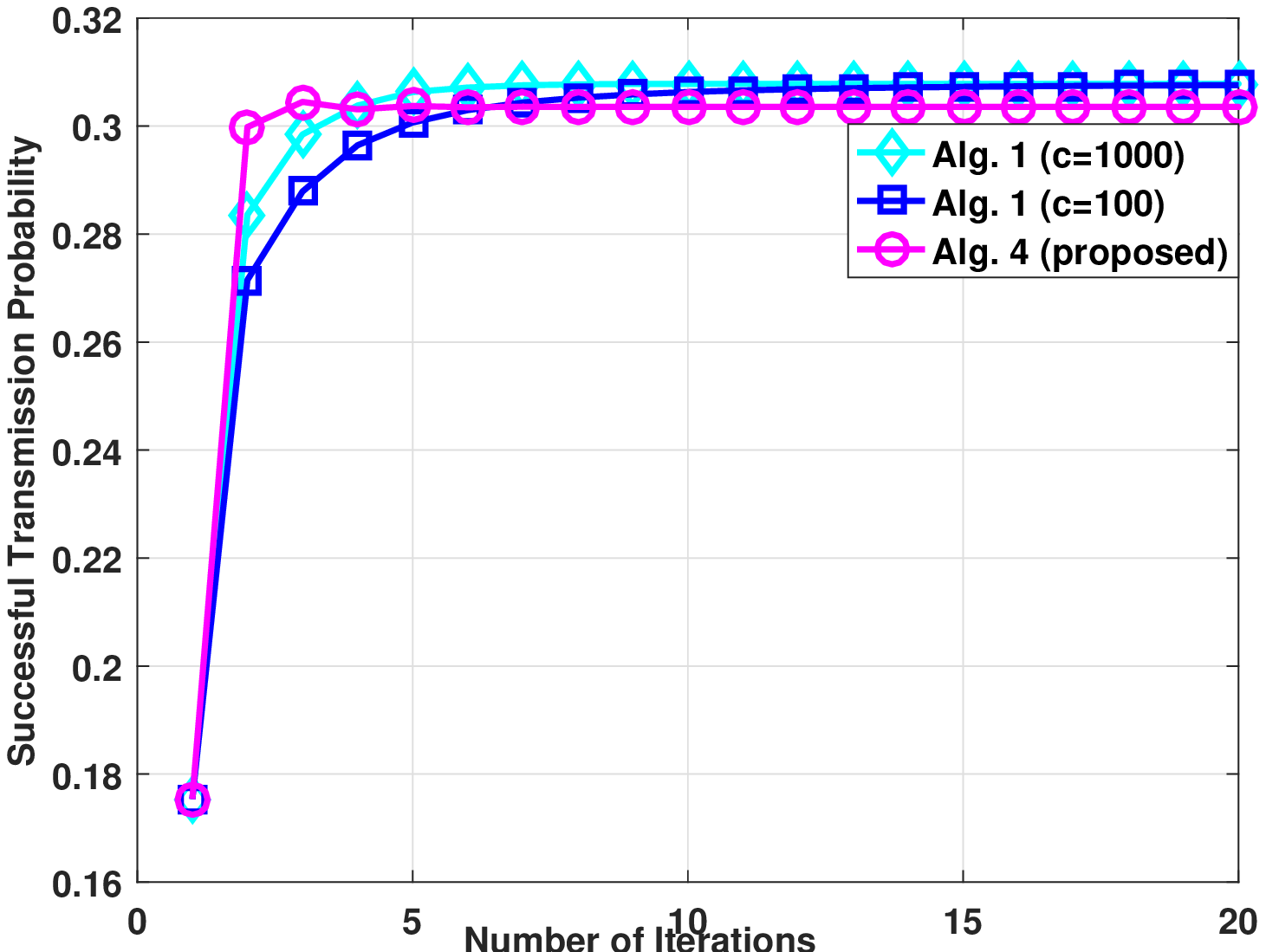}}}
\end{center}
\caption{\small{Successful transmission probability versus the number of iterations. The stepsize for Algorithm~\ref{alg:local-asymp-sym} is $ \epsilon(t)=\frac{c}{2+t^{0.55}}$. We choose the same initial point for all the algorithms shown in Fig.~\ref{fig:simulation-large-convergence}}}
\label{fig:simulation-large-convergence}
\end{figure}

\begin{figure}[t]
\begin{center}
\subfigure[\small{Cache size $K_2$ ($K_1$) at $\gamma=0.55$, $K_1=K_2+20$.}\label{fig:simulation-computaing-time-size}]
{\resizebox{7.7cm}{!}{\includegraphics{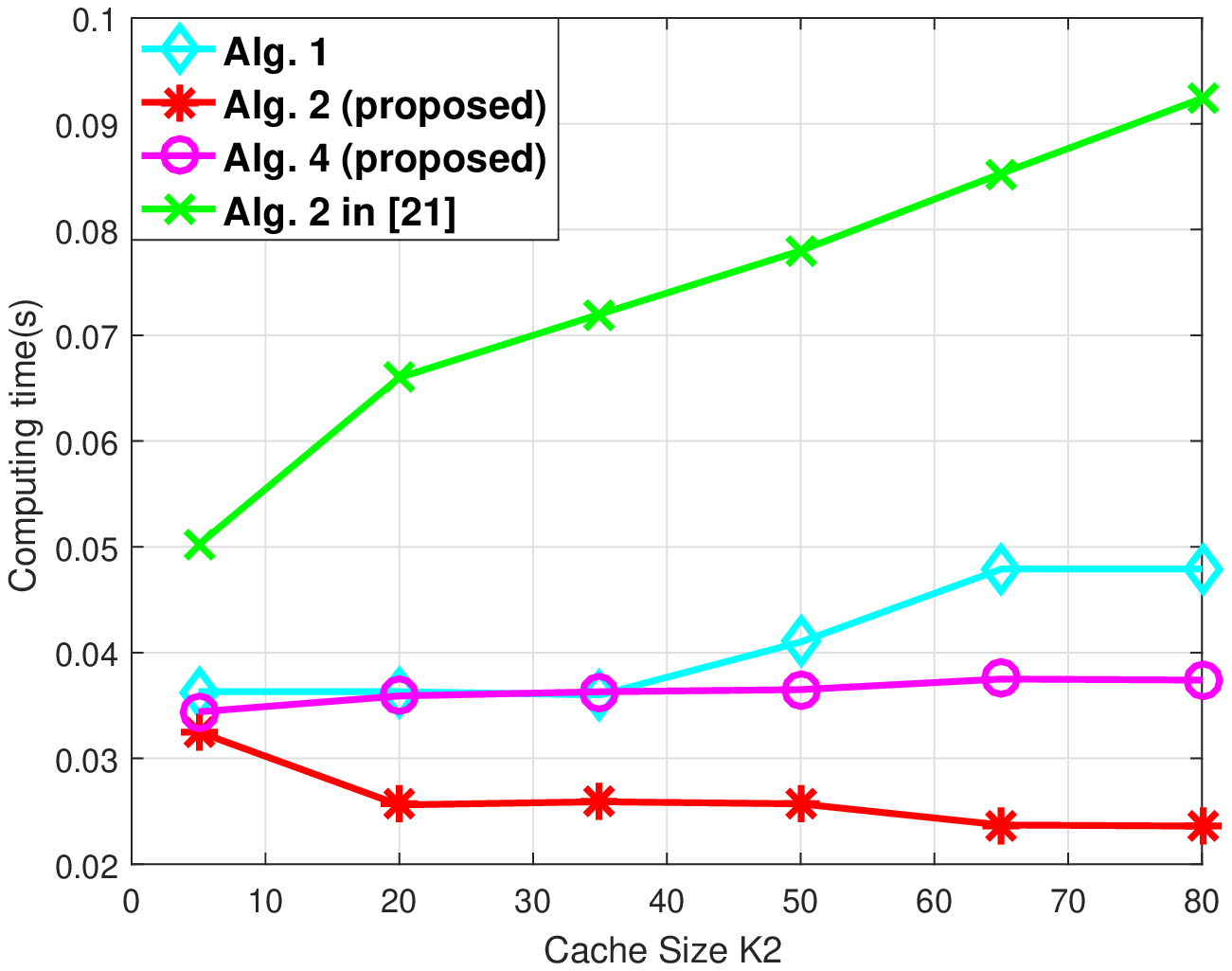}}}
\subfigure[\small{Zipf exponent $\gamma$ at $K_1=55$, $K_2=35$.}\label{fig:simulation-computaing-time-zipf-exponent}]
{\resizebox{7.7cm}{!}{\includegraphics{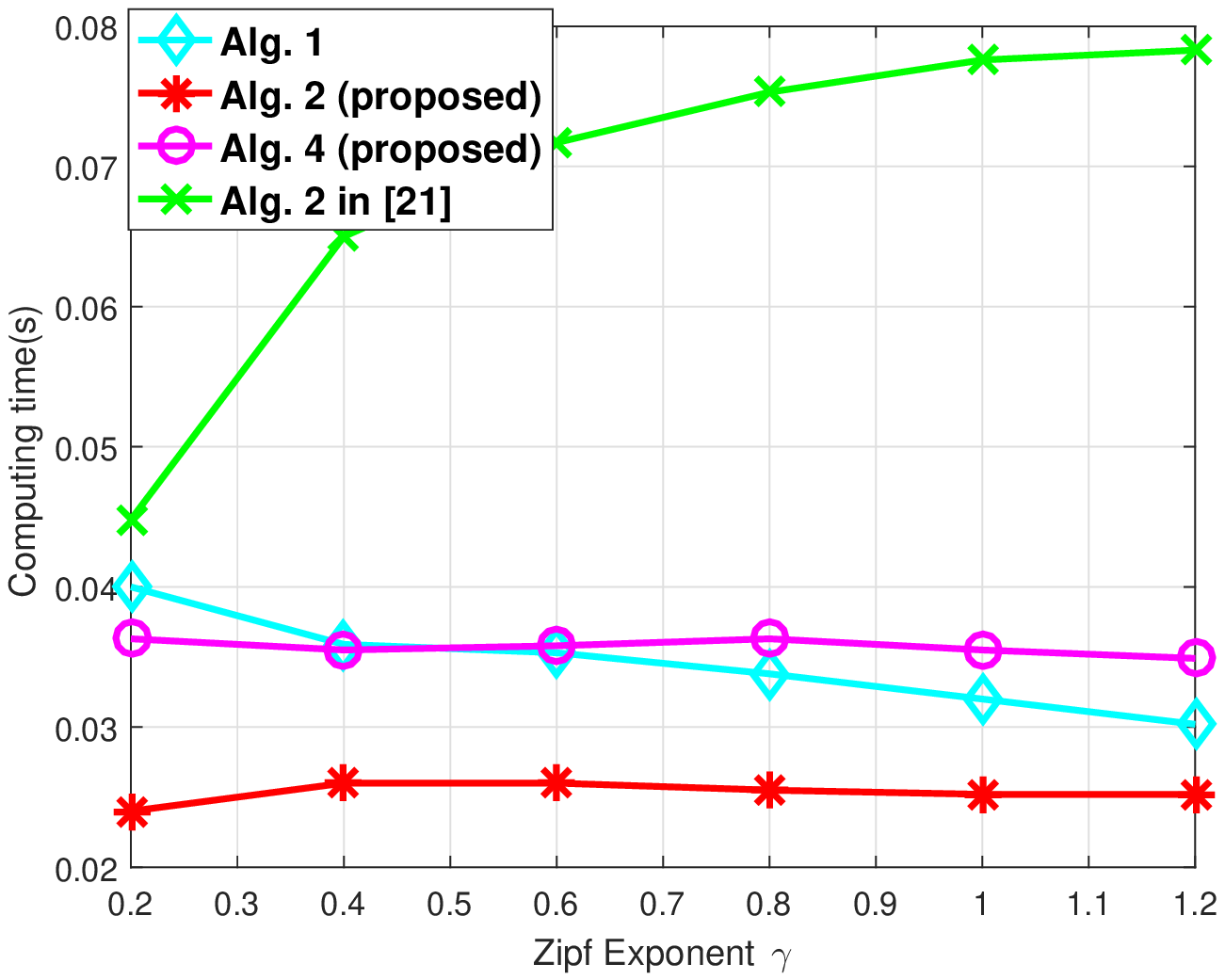}}}
\end{center}
\begin{center}
\subfigure[\small{
Cache size $K_2$ ($K_1$) at $\gamma=0.55$, $K_1=K_2$.}\label{fig:simulation-computaing-time-size-K}]
{\resizebox{7.7cm}{!}{\includegraphics{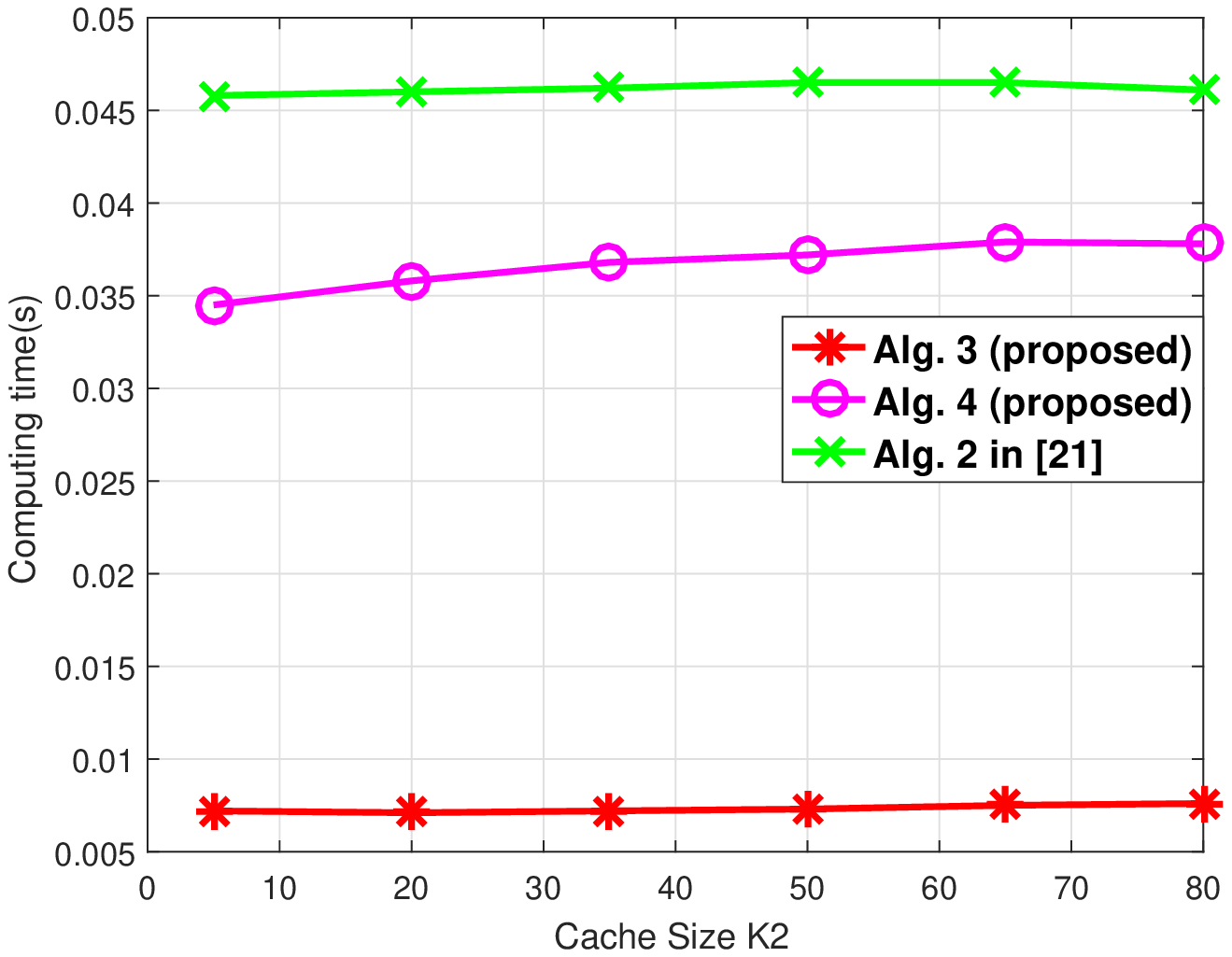}}}
\subfigure[\small{
Zipf exponent $\gamma$ at $K_1=K_2=35$.}\label{fig:simulation-computaing-time-zipf-exponent-K}]
{\resizebox{7.7cm}{!}{\includegraphics{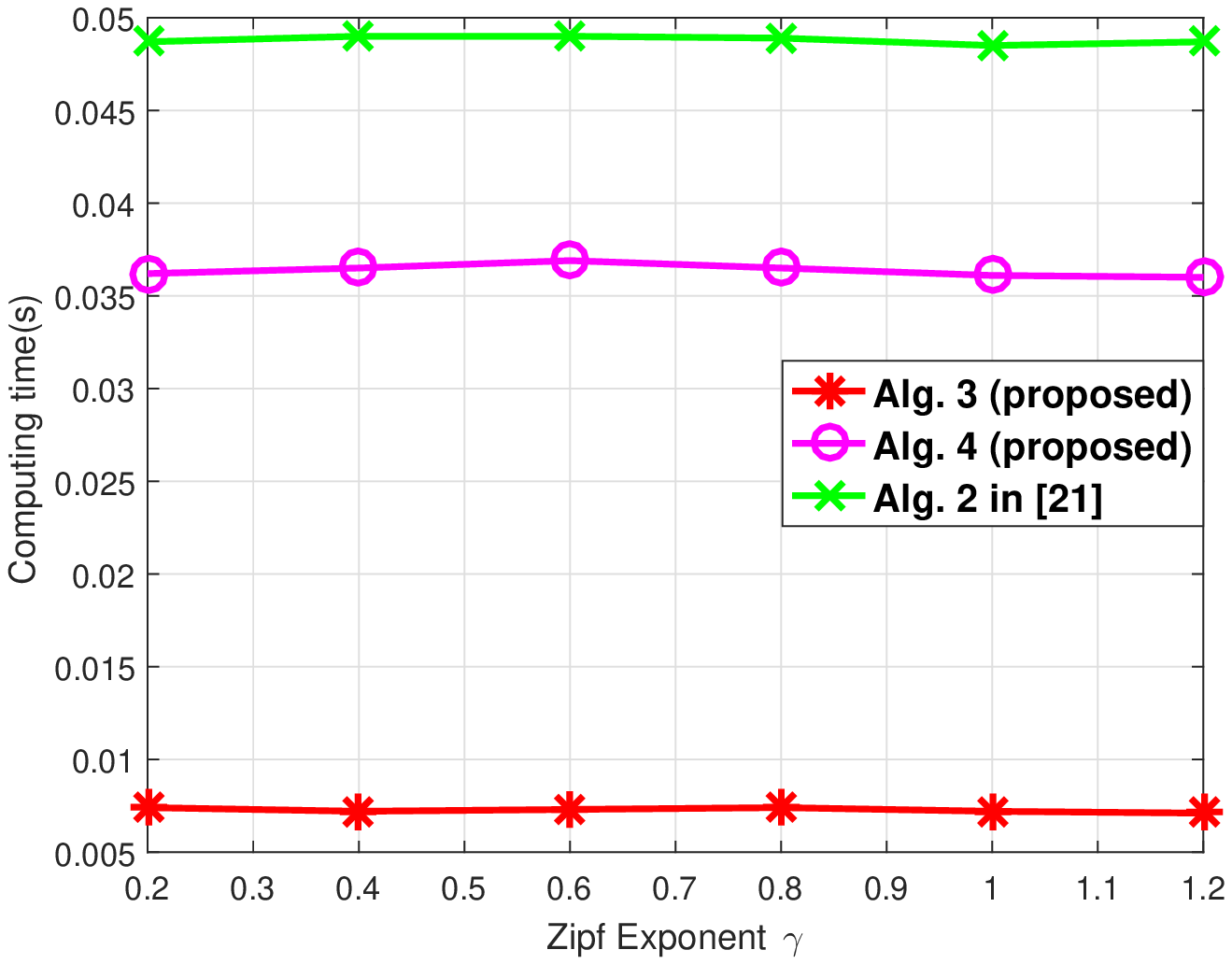}}}
\end{center}
\caption{\small{Computing time versus the cache size or Zipf exponent $\gamma$. The stepsize for Algorithm~\ref{alg:local-asymp-sym} is $ \epsilon(t)=\frac{c}{2+t^{0.55}}$. For Algorithm~\ref{alg:local-asymp-sym}, each point corresponds to the minimum computing time by choosing the optimal parameter $c\in \{500, 1000, 1500, 2000, 2500\}$.}}
\label{fig:simulation-computing-time}
\end{figure}

\subsection{Convergence and Complexity}
In this subsection, we show the convergence and complexity of the proposed algorithms. Fig.~\ref{fig:simulation-large-convergence} illustrates the successful transmission probability versus the number of iterations when $ K_1\neq K_2 $ and $ K_1=K_2 $. From Fig.~\ref{fig:simulation-large-convergence}, we can observe that the rate of convergence of Algorithm~\ref{alg:local-asymp-sym} is strongly dependent on the choices of stepsize $\epsilon(t)$. In addition, Algorithm~\ref{alg:BSUM} and Algorithm~\ref{alg:game-NE} have more robust convergence performance than Algorithm~\ref{alg:local-asymp-sym}, as they do not rely on a stepsize. Fig.~\ref{fig:simulation-computing-time} illustrates the computing time versus the cache size $K_j$ and the Zipf exponent $\gamma$ when $ K_1\neq K_2 $ and $ K_1=K_2 $. From  Fig.~\ref{fig:simulation-computing-time}, we can observe that the computing times of all the algorithms do not change much with $K_j$ or $\gamma$, and the computing times of the proposed algorithms are shorter than that of Algorithm 2 in \cite{cui2017analysis} which is to obtain an asymptotically optimal hybrid caching design. 
These observations demonstrate the advantage of the proposed algorithms in terms of complexity.

\begin{figure}[t]
\begin{center}
\subfigure[\small{
cache size $K_2$ ($K_1$) at $\gamma=0.55$, $K_1=K_2+20$.}\label{fig:simulation-large-size}]
{\resizebox{7.7cm}{!}{\includegraphics{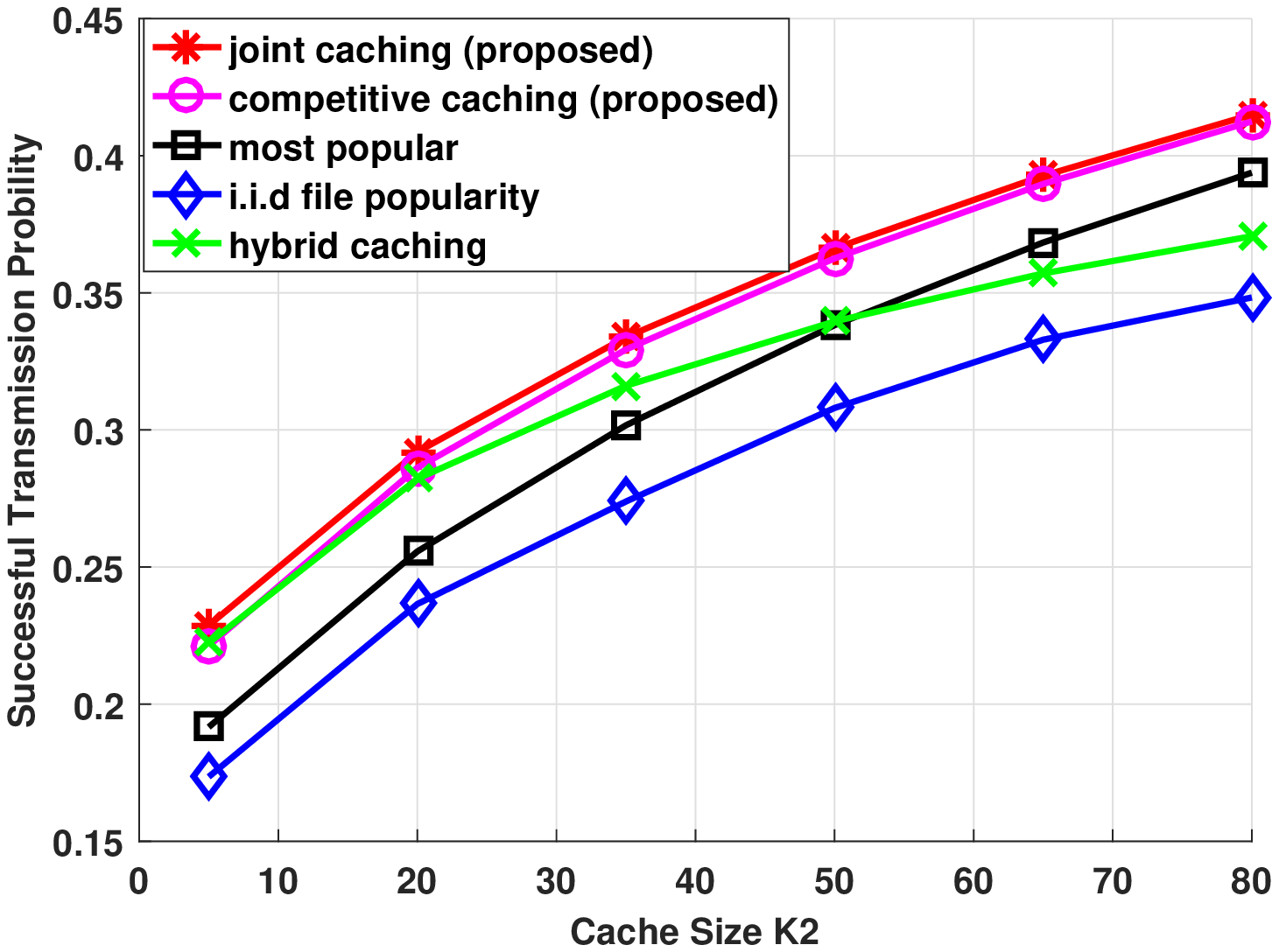}}}
\hspace{2mm}
\subfigure[\small{
Zipf exponent $\gamma$ at $K_1=55$, $K_2=35$.}\label{fig:simulation-large-zipf-exponent}]
{\resizebox{7.7cm}{!}{\includegraphics{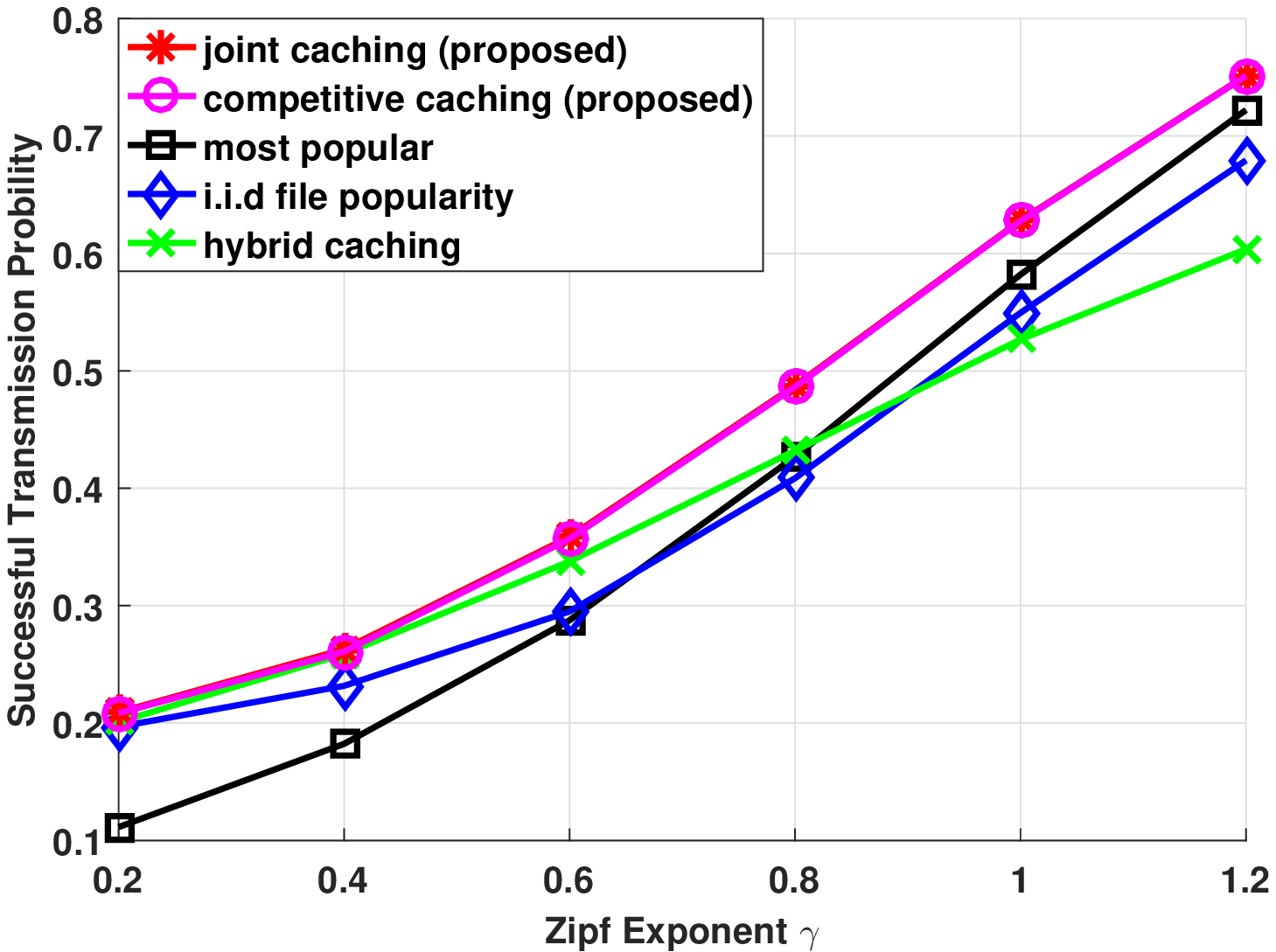}}}
\end{center}
\caption{\small{Successful transmission probability versus the cache size or Zipf exponent $\gamma$.}}
\label{fig:simulation-performance}
\end{figure}

\begin{figure}[t]
\begin{center}
\subfigure[$\lambda_{1}=1.3\times10^{-6}$, $\lambda_{2}=2.2\times10^{-6}$, $K_2=35$, $K_1=55$.
\label{fig:simulation-file-popularity-distribution1}]
{\resizebox{7.9cm}{!}{\includegraphics{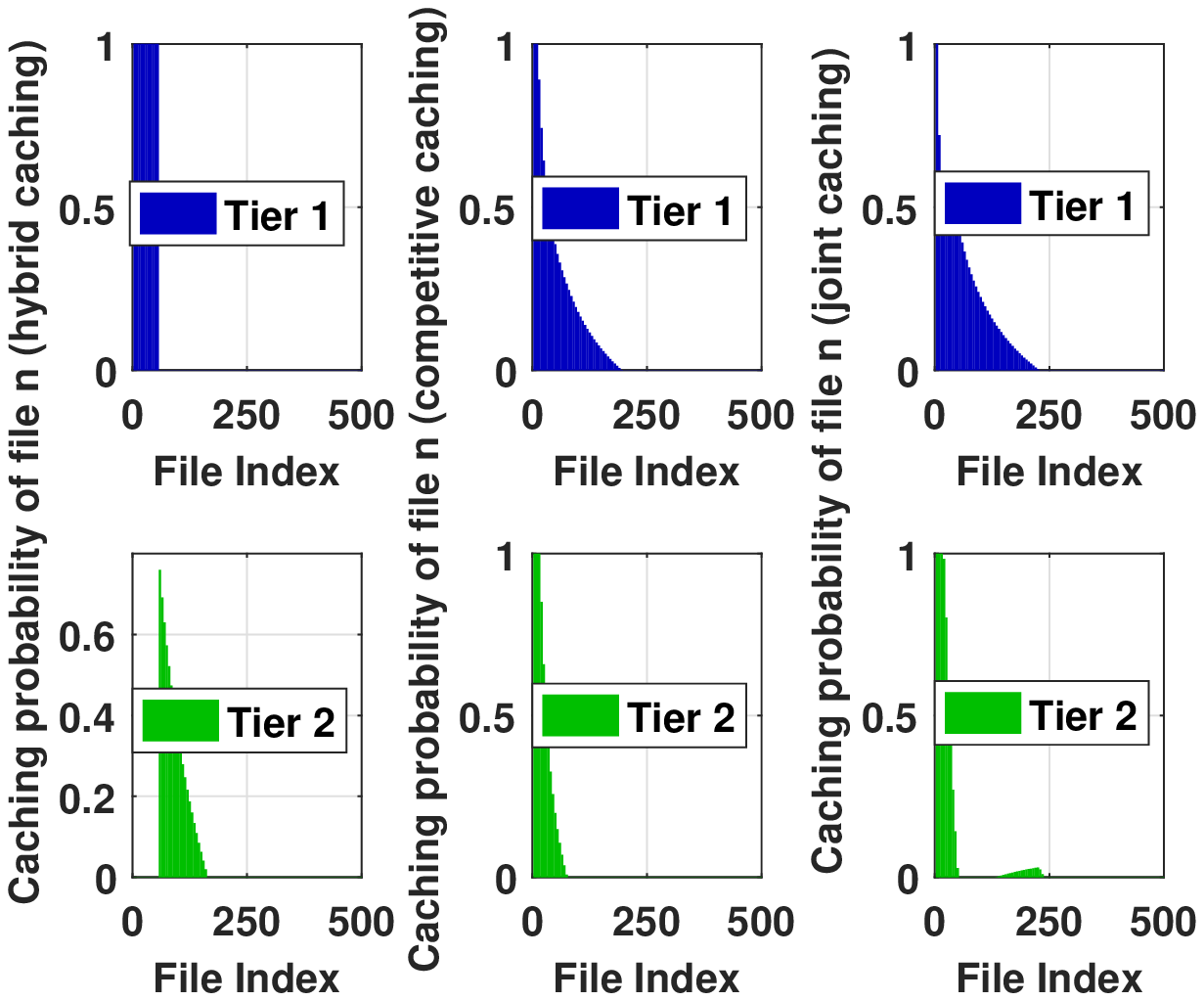}}}
\subfigure[$\lambda_{1}=9\times10^{-9}$, $\lambda_{2}=3.5\times10^{-6}$, $K_2=35$, $K_1=55$.
\label{fig:simulation-file-popularity-distribution2}]
{\resizebox{7.9cm}{!}{\includegraphics{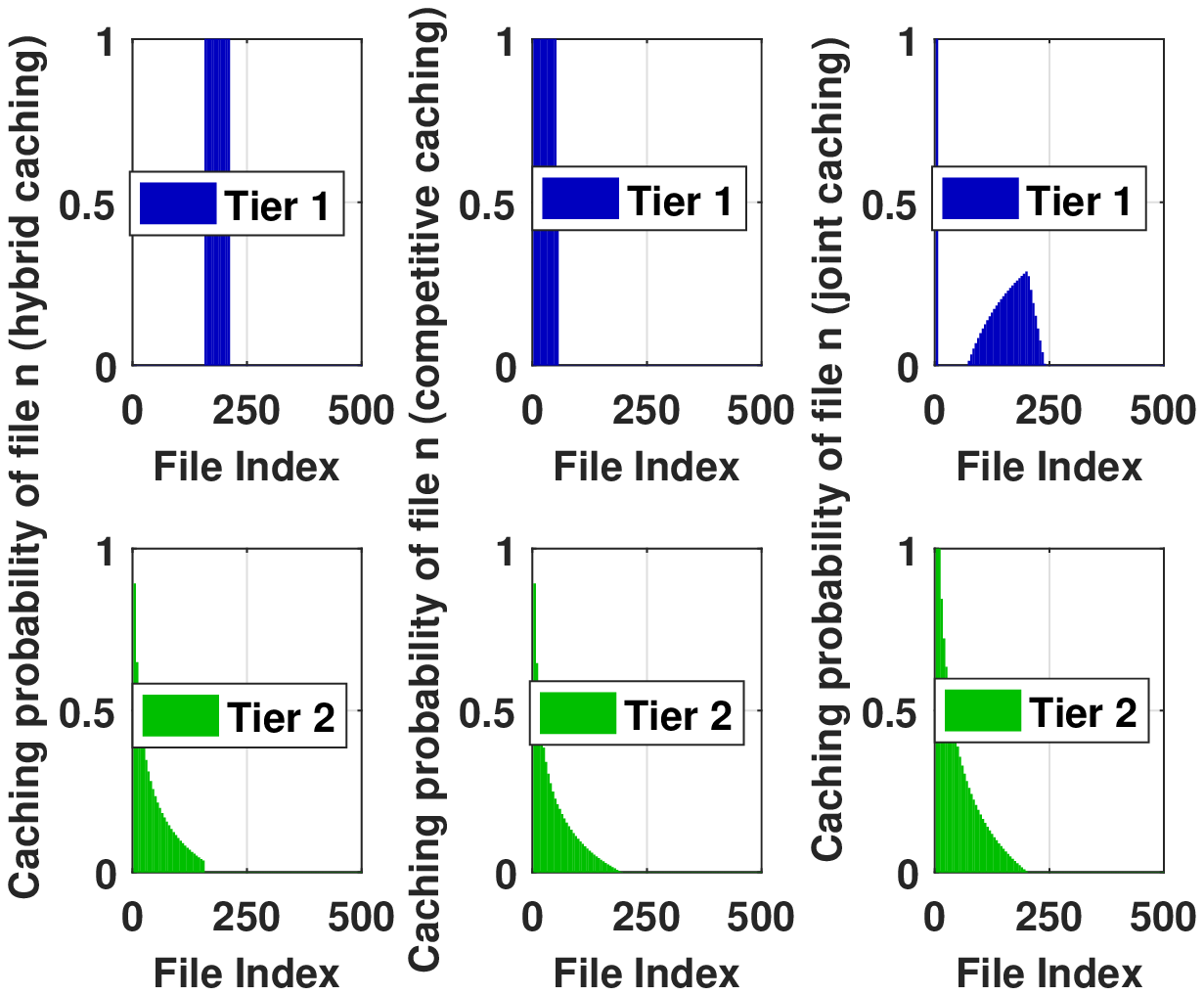}}}
\end{center}
\begin{center}
\subfigure[$\lambda_{1}=1.3\times10^{-6}$, $\lambda_{2}=2.2\times10^{-6}$, $K_1=K_2=35$.\label{fig:simulation-file-popularity-distribution1_K}]
{\resizebox{7.9cm}{!}{\includegraphics{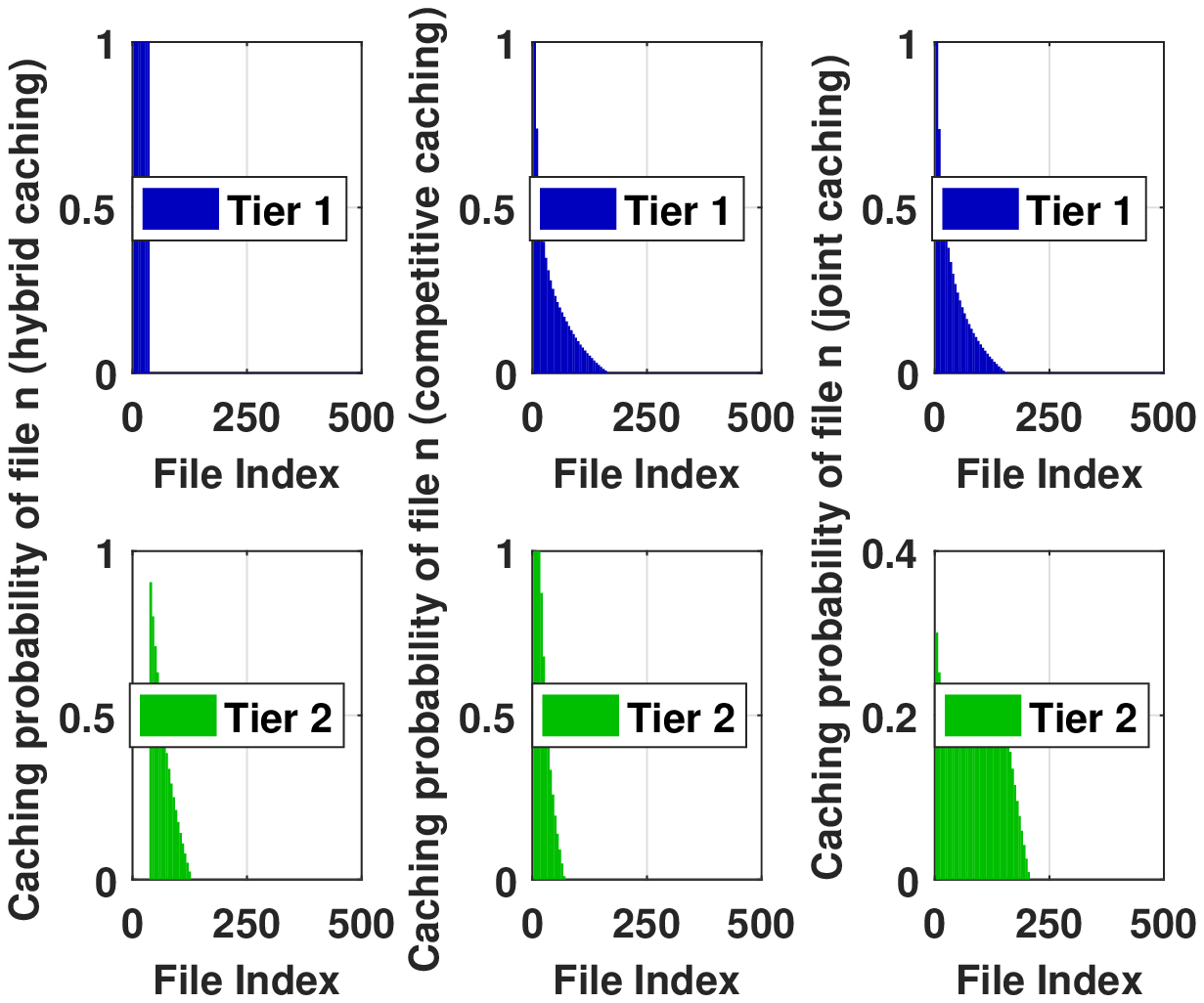}}}
\subfigure[$\lambda_{1}=9\times10^{-9}$, $\lambda_{2}=3.5\times10^{-6}$,$K_1=K_2=35$.
\label{fig:simulation-file-popularity-distribution2_K}]
{\resizebox{7.9cm}{!}{\includegraphics{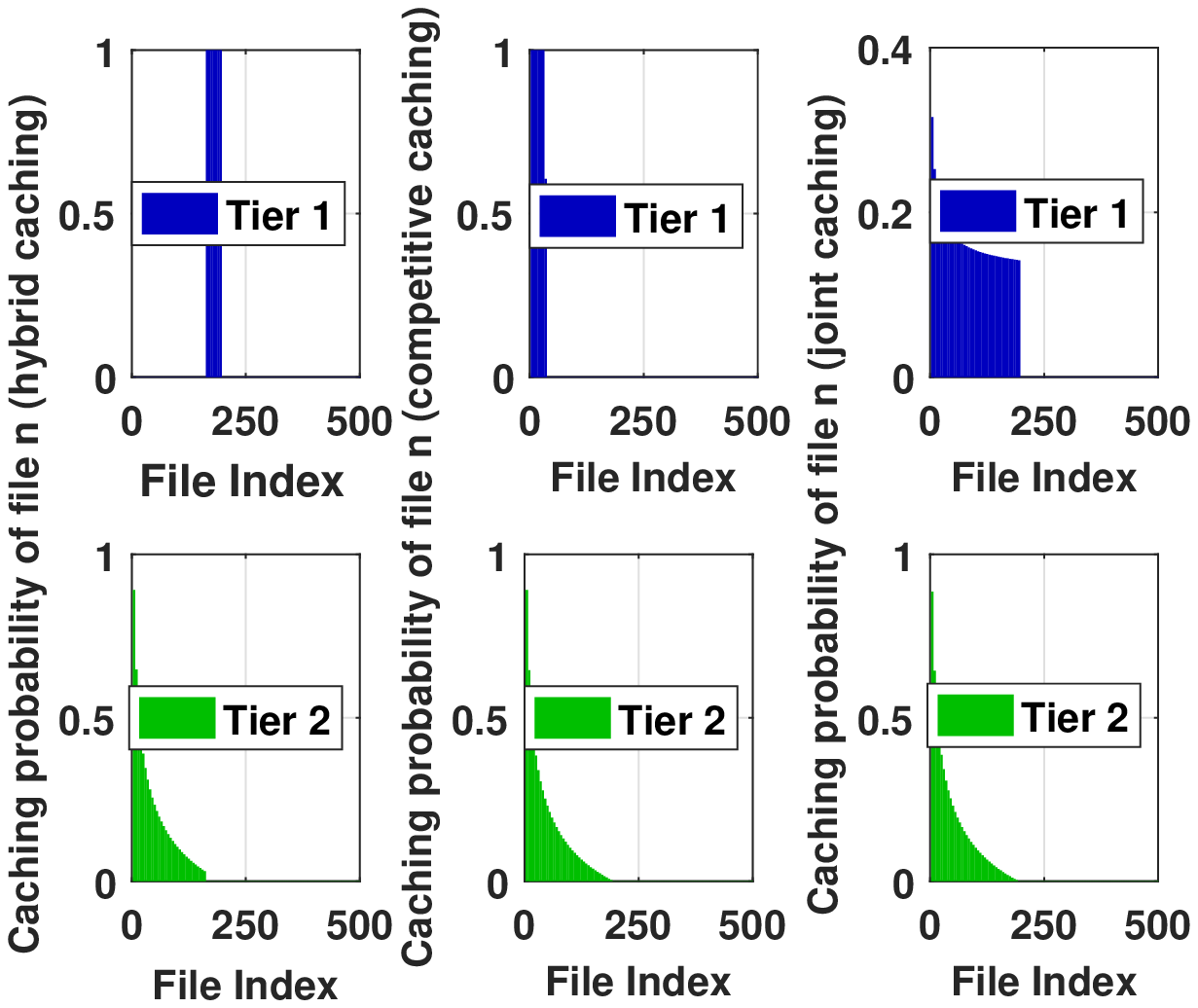}}}
\end{center}
\caption{\small{Caching probabilities for the files in $\{5, 10, \cdots, 500\}$ of joint caching design, competitive caching design and hybrid caching design at $\alpha=4$ and $\gamma=0.55$.}}
\label{fig:simu-file-popularity-distribution}
\end{figure}

\subsection{Successful Transmission Probabilities and Caching Probabilities}
In this subsection, we compare the successful transmission probabilities and caching probabilities of the proposed joint and competitive caching designs with those of three baselines. Baseline 1 (most popular) refers to the design in which each POA in tier $j$ stores the $K_j$ most popular files~\cite{EURASIP15Debbah,LiuYangICC16,Yang16}. Baseline 2 (i.i.d. file popularity) refers to the design in which each POA in tier $j$ randomly stores $K_j$ files, in an i.i.d. manner with file $n$ being selected with probability $a_n$~\cite{nagaraja2015caching}. Baseline 3 (hybrid caching) refers to the hybrid caching design obtained by Algrithm 2 in \cite{cui2017analysis}. The three baseline schemes also adopt the same multicasting scheme as in our design.
Fig.~\ref{fig:simulation-performance} illustrates the successful transmission probability versus the cache size $K_j$ and the Zipf exponent $\gamma$, respectively. From Fig.~\ref{fig:simulation-performance}, we can observe that as $K_j$ and $\gamma$ increase, the successful transmission probability of each scheme increases. We can also observe that the two proposed designs outperform all the three baseline schemes. In addition, we can see that when $K_j$ or $\gamma$ is large, the two proposed designs reduce to the most popular caching design. When $K_j$ or $\gamma$ is small, the two proposed designs perform similarly as the hybrid caching design. These observations show that the two proposed designs can well adapt to the changes of the system parameters and can wisely utilize storage resources.

Fig.~\ref{fig:simu-file-popularity-distribution} illustrates the caching probabilities for the proposed joint caching design, competitive caching design and the hybrid caching design. From Fig.~\ref{fig:simu-file-popularity-distribution}, we observe that under the proposed joint caching design and the hybrid caching design, when $\frac{\lambda_1}{\lambda_2}$ is above (below) some threshold, POAs of tier $ 1 $ (POAs of tier $ 2 $) cache the most popular files as they can offer relatively higher received powers. Recall that under the hybrid caching design, the files stored in the two tiers are non-overlapping, while the proposed joint caching design and competitive caching design allow a file to be stored in the two tiers. By comparing the caching probabilities under the three designs, we can see that the joint caching design and the competitive caching design offer much higher spatial file diversity, leading to higher successful transmission probabilities.

\section{Conclusion}
In this paper, we considered a  random caching  and multicasting scheme in a two-tier large-scale cache-enabled  wireless multicasting network, operated by a single operator or two different operators.
First, we  derived  tractable expressions for  the successful transmission probabilities in the general region and the asymptotic  region, respectively.
Then, we formulated  the optimal joint caching design problem  in the asymptotic region. We develop  an iterative algorithm, which is shown to converge to  a stationary point.
Next, we formulated the competitive caching design game in the asymptotic region, obtained the unique NE of the game and developed an iterative algorithm, which is shown to converge to  the NE under a mild condition.
Finally, by numerical simulations, we showed that the two proposed designs achieve significant gains over  existing schemes, in terms of  successful transmission probability and complexity.

\section*{Appendix A: Proof of Lemma~\ref{Lem:pmf-K}}

Let random variable $Y_{j,m,n,i}\in\{0,1\}$ denote whether file $m\in \mathcal N_{j,i}\setminus \{n\}$ is requested from $\ell_{0}$ when $\ell_{0}$ contains combination $i\in \mathcal I_{j,n}$. When $\ell_{0}$ contains combination $i\in \mathcal I_{j,n}$, we have $K_{j,n,0}=1+\sum_{m\in \mathcal N_{j,i,-n}} Y_{j,m,n,i}$.
For analytical tractability, as in~\cite{cui2016analysis}, assume $ Y_{j,m,n,i} $, $ m\in\mathcal N_{j,i}\setminus \{n\} $ are independent.
By Appendix~C of~\cite{cui2016analysis}, we have
\begin{align}
&\Pr \left[K_{j,n,0}=k\ |\ j_0=j\right]\nonumber\\
\approx&\sum_{i\in \mathcal I_{j,n}}\frac{p_{j,i}}{T_{j,n}}\sum_{ \mathcal X\in  \left\{\mathcal S \subseteq \mathcal N_{j,i,-n}: |\mathcal S|=k-1\right\} }\prod_{m\in \mathcal X}(1-\Pr[Y_{j,m,n,i}=0])\prod_{m\in \mathcal N_{j,i,-n}\setminus\mathcal X}\Pr[Y_{j,m,n,i}=0].\label{eqn:load-pmf-in-proof}
\end{align}
Similar to Appendix~B in~\cite{singh2013offloading}, we have
\begin{align}
\Pr[Y_{j,m,n,i}=0]
\approx\left(1+\frac{a_m\lambda_uA_{j,m}\left(T_{j,m},T_{\overline j,m}\right)}{3.5T_{j,m}\lambda_j}\right)^{-3.5}
=b_{j,m}.\label{eqn:voronoi-cell-prob}
\end{align}
By substituting~\eqref{eqn:voronoi-cell-prob} into~\eqref{eqn:load-pmf-in-proof}, we can prove Lemma~\ref{Lem:pmf-K}.

\section*{Appendix B: Proof of  Corollary~\ref{Cor:asym-perf-v2}}

When $\frac{P_1}{N_0}\to \infty$ and $\frac{P_2}{N_0}\to \infty$, we have $\exp\left(-\left (2^{\frac{k\tau}{W}} - 1\right ) d^{\alpha}\frac{N_{0}}{P_1}\right)\to1$ and\\ $\exp\left(-\left (2^{\frac{k\tau}{W}} - 1 \right ) d^{\alpha}\frac{N_{0}}{P_2}\right )\to1$. When $\lambda_u\to \infty$,  we have $K_{j,n,0} \to K_j$ in distribution. 
Noting that $\int_{0}^{\infty}d\exp\left(-cd^{2}\right){\rm d}d=\frac{1}{2c}$, we can solve integrals in \eqref{eqn:f-j-k}. Thus, we can prove Corollary~\ref{Cor:asym-perf-v2}.


\section*{Appendix C: Proof of Theorem~\ref{Thm:BSUM-cvg}}

We show that  the conditions in  Theorem 2 (a) of \cite{razaviyayn2013unified} hold. 
i) By noting that \eqref{eqn:BSUM-tight} and \eqref{eqn:BSUM-lower-bound}  hold and $ g_j\left(\mathbf T_{j};\mathbf T_1(t),\mathbf T_2(t)\right) $ is continuous and differentiable, we know that Assumption~2 in~\cite{razaviyayn2013unified} is satisfied and $ q_{\infty}\left(\mathbf T_1,\mathbf T_2\right) $ is regular at any point in $ \mathcal T_1\times\mathcal T_2 $\cite{razaviyayn2013unified}.
ii) Since $ g_j\left(\mathbf T_{j};\mathbf T_1(t),\mathbf T_2(t)\right) $ is strictly concave on $ \mathcal T_j $ for any given $ \left(\mathbf T_1(t),\mathbf T_2(t)\right)\in\mathcal T_1\times\mathcal T_2 $, Problem~\ref{prob:BSUM-subprob} has a unique solution for any given $ \left(\mathbf T_1(t),\mathbf T_2(t)\right)\in\mathcal T_1\times\mathcal T_2 $. 
Therefore, by Theorem 2 (a) in\cite{razaviyayn2013unified}, we can prove Theorem~\ref{Thm:BSUM-cvg}. 

\section*{Appendix D: Proof of Lemma~\ref{Lem:NE-exis-uniq}}

First, we use Proposition 20.3 in \cite{osborne1994course} to prove the existence of NE of game $ \mathcal G $. It is obvious that for all $ j=1,2 $,  the set of admissible strategies $ \mathcal T_{j} $ is nonempty, compact and convex, and the utility function $ q_{j,\infty}(\mathbf T_j,\mathbf T_{\overline j}) $ is a  continuous function of $ \left(\mathbf T_1,\mathbf T_{2}\right) $. Since $ q_{j,\infty}(\mathbf T_j,\mathbf T_{\overline j}) $ is strictly concave on $ \mathcal T_{j} $ for any given  $ \mathbf T_{\overline j}\in\mathcal T_{\overline j} $, then  $ q_{j,\infty}(\mathbf T_j,\mathbf T_{\overline j}) $ is  quasi-concave on $ \mathcal T_{j} $ for any given  $ \mathbf T_{\overline j}\in\mathcal T_{\overline j} $. Thus, by Proposition 20.3 in \cite{osborne1994course}, we know that there exists at least one NE of game $ \mathcal G $. 
Next, we prove the uniqueness of NE by Theorem 2 in \cite{rosen1965existence}. By the first-order strict concavity condition, we know that a strictly concave function must be diagonally strictly concave\cite{rosen1965existence}. Thus, by Theorem 2 in \cite{rosen1965existence}, we know that there exists a unique NE of game $ \mathcal G $.

\section*{Appendix E: Proof of Theorem~\ref{Thm:NE-alg-cvg}}

We prove the convergence of Algorithm~\ref{alg:game-NE} by verifying the conditions in Theorem 1 of\cite{li1987distributed}. i) It can be easily seen that for all $ j=1,2 $, $ q_{j,\infty}(\mathbf T_j,\mathbf T_{\overline j}) $ is strictly concave on $ \mathcal T_{j} $ for any given $ \mathbf T_{\overline j}\in\mathcal T_{\overline j} $ and is second-order Fr\'{e}chet differentiable\cite{luenberger1969optimization}.  
ii) By Lemma~\ref{Lem:NE-alg-subprob-opt}, we know that there  exists an optimal solution to Problem~\ref{prob:NE-alg-subprob} for any given $ \mathbf T_{\overline j}(t)\in\mathcal T_{\overline j} $. iii) To guarantee the convergence of Algorithm~\ref{alg:game-NE}, it remains to show that
\begin{align}
\Bigg{\lVert} &(\nabla^2_{\scriptscriptstyle{\mathbf T_j^2}}q_{j,\infty}\left(\mathbf T_j,\mathbf T_{\overline j}\right))^{-1} \nabla^2_{\scriptscriptstyle{\mathbf T_j\mathbf T_{\overline j}}}q_{j,\infty}\left(\mathbf T_j,\mathbf T_{\overline j}\right)\Big{|}_{\scriptscriptstyle{\mathbf T_j = \mathbf T_j(t+2),\mathbf T_{\overline j} = \mathbf T_{\overline j}(t+1)}} \nonumber\\
&\times
(\nabla^2_{\scriptscriptstyle{\mathbf T_{\overline j}^2}}q_{\overline j,\infty}\left(\mathbf T_{\overline j},\mathbf T_j\right))^{-1} \nabla^2_{\scriptscriptstyle{\mathbf T_{\overline j}\mathbf T_j}}q_{\overline j,\infty}\left(\mathbf T_{\overline j},\mathbf T_j\right)\Big{|}_{\scriptscriptstyle{\mathbf T_j=\mathbf T_j(t),\mathbf T_{\overline j}=\mathbf T_{\overline j}(t+1)}} \Bigg{\rVert}_{2} = \zeta < 1, 
\end{align}
holds for any $ \mathbf T_j(t)\in\mathcal T_j $ and $ ((t+1)\text{ mod }2)+1 = \overline j $. 
Here, $ \lVert\,\cdot\,\rVert_{2} $  denotes the spectral norm~\cite{luenberger1969optimization}. 
By Corollary~\ref{Cor:asym-perf-v2}, we have
\begin{align}
\Bigg{\lVert} &(\nabla^2_{\scriptscriptstyle{\mathbf T_j^2}}q_{j,\infty}\left(\mathbf T_j,\mathbf T_{\overline j}\right))^{-1} \nabla^2_{\scriptscriptstyle{\mathbf T_j\mathbf T_{\overline j}}}q_{j,\infty}\left(\mathbf T_j,\mathbf T_{\overline j}\right)\Big{|}_{\scriptscriptstyle{\mathbf T_j = \mathbf T_j(t+2),\mathbf T_{\overline j} = \mathbf T_{\overline j}(t+1)}} \nonumber\\
&\times
(\nabla^2_{\scriptscriptstyle{\mathbf T_{\overline j}^2}}q_{\overline j,\infty}\left(\mathbf T_{\overline j},\mathbf T_j\right))^{-1} \nabla^2_{\scriptscriptstyle{\mathbf T_{\overline j}\mathbf T_j}}q_{\overline j,\infty}\left(\mathbf T_{\overline j},\mathbf T_j\right)\Big{|}_{\scriptscriptstyle{\mathbf T_j=\mathbf T_j(t),\mathbf T_{\overline j}=\mathbf T_{\overline j}(t+1)}} \Bigg{\rVert}_{2} \nonumber\\
\eqla&\Bigg{\lVert} \text{diag}\left (\left (\frac{1}{4} \left(1-\frac{\theta_{1,K_j}T_{j,n}(t+2)}{\theta_{2,j,K_j}T_{\overline j,n}(t+1) + \theta_{3,j,K_j}}\right) \left(1-\frac{\theta_{1,K_{\overline j}}T_{\overline j,n}(t+1)}{\theta_{2,\overline j,K_{\overline j}}T_{j,n}(t) + \theta_{3,\overline j,K_{\overline j}}}\right)\right )_{n\in\mathcal N}\right )\Bigg{\rVert}_{2} \nonumber\\
\eqlb&\max_{n\in\mathcal N}\Bigg{\lvert} \frac{1}{4} \left(1-\frac{\theta_{1,K_j}T_{j,n}(t+2)}{\theta_{2,j,K_j}T_{\overline j,n}(t+1) + \theta_{3,j,K_j}}\right) \left(1-\frac{\theta_{1,K_{\overline j}}T_{\overline j,n}(t+1)}{\theta_{2,\overline j,K_{\overline j}}T_{j,n}(t) + \theta_{3,\overline j,K_{\overline j}}}\right) \Bigg{\rvert} \nonumber\allowdisplaybreaks\\
\stackrel{(c)}{\leq}&\, \frac{1}{4}  \max_{n\in\mathcal N}\Bigg{\lvert}  1-\frac{\theta_{1,K_j}T_{j,n}(t+2)}{\theta_{2,j,K_j}T_{\overline j,n}(t+1) + \theta_{3,j,K_j}} \Bigg{\rvert} \times \max_{n\in\mathcal N}\Bigg{\lvert} 1-\frac{\theta_{1,K_{\overline j}}T_{\overline j,n}(t+1)}{\theta_{2,\overline j,K_{\overline j}}T_{j,n}(t) + \theta_{3,\overline j,K_{\overline j}}} \Bigg{\rvert} \nonumber\\
\stackrel{(d)}{\leq}&\, \frac{1}{4}  \max\left\{1,\ \Bigg{\lvert}{1-\frac{\theta_{1,K_j}}{\theta_{3,j,K_j}}}\Bigg{\lvert}\right\} \max\left\{1,\ \Bigg{\lvert}{1-\frac{\theta_{1,K_{\overline j}}}{\theta_{3,\overline j,K_{\overline j}}}}\Bigg{\lvert}\right\} < 1,
\end{align}
where $ (a) $ is obtained by the definition of second-order derivative, $ (b) $ is obtained by the definition of spectral norm, $ (c) $ is obtained based on the formula $ \max\big{\lvert}{x_ny_n}\big{\rvert} \leq \max\big{\lvert}{x_n}\big{\rvert}\cdot\max\big{\lvert}{y_n}\big{\rvert} $, $ n\in\mathcal N $ and $ (d) $ is obtained due to
\begin{align}
1-\frac{\theta_{1,K_j}}{\theta_{3,j,K_j}} &\leq 1-\frac{\theta_{1,K_j}T_{j,n}(t+2)}{\theta_{2,j,K_j}T_{\overline j,n}(t+1) + \theta_{3,j,K_j}} \leq 1,\\
1-\frac{\theta_{1,K_{\overline j}}}{\theta_{3,\overline j,K_{\overline j}}} &\leq 1-\frac{\theta_{1,K_{\overline j}}T_{\overline j,n}(t+1)}{\theta_{2,\overline j,K_{\overline j}}T_{j,n}(t) + \theta_{3,\overline j,K_{\overline j}}} \leq 1. 
\end{align}
Therefore, by Theorem 1 of \cite{razaviyayn2013unified}, we can prove Theorem~\ref{Thm:NE-alg-cvg}.

%
%


\end{document}